\DeclareUnicodeCharacter{2009}{\,}
\documentclass[%
reprint,
superscriptaddress,
nofootinbib,
aps,longbibliography,
pra,
amsmath,amssymb
]{revtex4-1}

\usepackage[toc,page]{appendix} 

\newif\ifarxiv 
\arxivtrue

\usepackage[caption=false]{subfig}

\usepackage{array} 
\usepackage{xparse}
\usepackage[english]{babel}
\usepackage[utf8]{inputenc}
\usepackage{setspace}
\usepackage{pgfplots}
\usepgfplotslibrary{fillbetween} 
\pgfplotsset{compat=newest}
\usepackage{graphicx}
\graphicspath{ {./figures/} }
\usepackage[caption=false]{subfig} 
\usepackage{amsmath}
\usepackage{amsthm}
\usepackage{amssymb}
\usepackage{braket}
\usepackage{dsfont}
\usepackage{xcolor}
\usepackage{mathrsfs} 
\usepackage{mathtools} 
\usepackage{tikz}
\usetikzlibrary{calc}
\usetikzlibrary{arrows}
\usepackage{tcolorbox}
\usepackage{enumitem}
\newlist{steps}{enumerate}{1}
\setlist[steps, 1]{label = Step \arabic*:}
\usepackage[colorinlistoftodos,prependcaption]{todonotes}
\usepackage{booktabs}
\usepackage[utf8]{inputenc}
\usepackage{thmtools,thm-restate}
\usepackage{hyperref} 
\usepackage{xurl} 
\hypersetup{
	colorlinks = true,
	citecolor= blue,
	urlcolor= blue,
	linkcolor = blue
}
\usepackage{algorithm}
\usepackage[noend]{algpseudocode}

\newtheorem{theorem}{Theorem}

\newtheorem{definition}{Definition}
\newtheorem{lemma}{Lemma}
\newtheorem{prop}{Proposition}
\newtheorem{rem}{Remark}

\newcounter{algsubstate}


\usepackage{bm}
\newcommand{\ketbra}[2]{\ket{#1}\!\bra{#2}}

\let\inner\relax
\NewDocumentCommand\inner{mg}{%
	\ensuremath{\left\langle #1 | \IfNoValueTF{#2}{#1}{#2}\right\rangle}%
}

\usepackage[T3,T1]{fontenc}
\DeclareSymbolFont{tipa}{T3}{cmr}{m}{n}
\DeclareMathAccent{\invbreve}{\mathalpha}{tipa}{16}
\allowdisplaybreaks

\begin{document}
	
	\title{Certified Randomness Generation against Quantum Adversaries via Energy Constraints}
	

\author{Thomas A. Hahn}
 \affiliation{%
 The Center for Quantum Science and Technology, \\ Department of Physics of Complex Systems, Weizmann Institute of Science, Rehovot, Israel
 }
 \affiliation{%
Institute of Computer and Communication Sciences, École Polytechnique Fédérale de Lausanne (EPFL), Lausanne CH-1015, Switzerland
 } 
 \affiliation{LMU Munich $\&$ Munich Center for Quantum Science and Technology (MCQST), Munich, Germany}

    \author{Efrat Gerchkovitz}
     \affiliation{%
 The Center for Quantum Science and Technology, \\ Department of Physics of Complex Systems, Weizmann Institute of Science, Rehovot, Israel
 }
	\author{Ester Yehuda}
	\affiliation{%
 The Center for Quantum Science and Technology, \\ Department of Physics of Complex Systems, Weizmann Institute of Science, Rehovot, Israel
 }
	\author{Rotem Arnon}
 \affiliation{%
 The Center for Quantum Science and Technology, \\ Department of Physics of Complex Systems, Weizmann Institute of Science, Rehovot, Israel
 }
 \affiliation{%
Institute of Computer and Communication Sciences, École Polytechnique Fédérale de Lausanne (EPFL), Lausanne CH-1015, Switzerland
 } 


	\date{\today}

	\begin{abstract} 
        For quantum random number generators (QRNGs), strong security guarantees can be obtained under various assumptions. The most well-known assumptions are the non-signalling assumption (needed for device-independent (DI) QRNG) and post-quantum computational assumptions (needed for computationally bounded single-prover QRNG protocols). In this work, we study a third model: a simple prepare-and-measure QRNG whose security relies on physically motivated energy constraints. This physically motivated set-up allows for easier implementations of QRNGs and can supply better rates. 

So far, previous works that studied the energy-constrained model either assumed that the adversary is classical in their security proof or supplied relatively weak bounds on the certified randomness, by bounding the conditional min-entropy.  
In this work, we prove a new, essentially tight, lower bound on the single-round von Neumann entropy against adversaries holding both classical and quantum side-information. Our work reveals that working with the von Neumann entropy, rather than the min-entropy, in the quantum case can lead to a major improvement, nearly matching the classical bound in many cases. We further study the robustness of our entropy bound.
Since the single-round von Neumann entropy is the main quantity of interest when characterizing the key rate of QRNGs, our bound establishes the main ingredient for rigorous security proofs against general quantum attacks in this framework.

	\end{abstract}
	
	\maketitle


	\section{Introduction}
    Quantum theory allows for the generation of intrinsic randomness: even with complete knowledge of a system’s preparation, the outcome of an individual measurement can remain fundamentally unpredictable. This forms the basis for quantum random number generators (QRNGs), which exploit this phenomenon to create uniformly random bits, even in the presence of external eavesdroppers~\cite{Ac_n_2016,Herrero_Collantes_2017}. Practical QRNGs, however, are susceptible to manufacturing imperfections, implementation flaws, or even malicious interference~\cite{Cao_2016,Haw_2015}. The security of a QRNG has to therefore take such effects into account and cannot be inferred solely from the quantum nature of its ideal operating state~\cite{Ac_n_2016,Herrero_Collantes_2017,brask2026quantumcorrelationsprepareandmeasurescenarios,Pauwels_2025}.
Two well-known models under which such randomness certification is possible (even against quantum adversaries) are given by device-independent (DI)  protocols~\cite{acin2007device,pironio2009deviceindependent,arnon2019simple,Arnon-Friedman:2018aa} and single-prover interactive protocols under computational assumptions~\cite{Metger_2021,e27080772,Merkulov_2026,8555116}. 

DI cryptographic protocols fundamentally require multiple devices which do not signal, i.e., leak information, to each other during their implementation phase. Other than this assumption, DI protocols 
treat these quantum devices as uncharacterized black boxes that simply obey the laws of quantum mechanics~\cite{Arnon_Friedman_2020,tan2024prospectsdeviceindependentquantumkey}.
 The produced randomness, i.e., the (raw) 
key, is then typically certified via Bell inequalities~\cite{arnon2019simple,Arnon-Friedman:2018aa}. However, while DI protocols have been implemented in practice~\cite{dirng_expr1,dirng_expr2,dira_expr,diqkd_expr1,diqkd_expr2,diqkd_expr3,Lu_2026}, the resulting finite-size key rate is typically very low; see also~\cite{arnon2019simple,Arnon-Friedman:2018aa,hahn2025analyticrenyientropybounds}. 
Computationally bounded single-prover QRNG protocols~\cite{8555116,Metger_2021,e27080772,Merkulov_2026}, on the other hand, can certify randomness on single devices by replacing the non-signaling assumption of DI protocols with post-quantum computational assumptions, such as relying on the hardness of Learning With Errors (LWE)~\cite{10.1145/1568318.1568324,10.1561/0400000074}. Also for computationally bounded single-prover QRNG protocols, their finite-size key rates remain very low~\cite{e27080772,Merkulov_2026}.

Such practical drawbacks have led to the development of so-called semi-DI security. 
Unlike fully DI protocols,  semi-DI protocols rely on a limited set of additional assumptions to ensure security. These assumptions  impose additional physical constraints on the underlying quantum systems~\cite{brask2026quantumcorrelationsprepareandmeasurescenarios,Pauwels_2025,Brask_2017,Wang_2019,Li_2011,Lunghi_2015,PhysRevA.85.052308,Mironowicz_2016,Tavakoli_2020,Tavakoli_2022,Flatt_2022,Roch_i_Carceller_2022,Tavakoli_2021,Jones_2026,Van_Himbeeck_2017,vanhimbeeck2019correlations,Roch_i_Carceller_2025,senno2021,DAvino:2025ymd,Roch_i_Carceller_2026,ranu2026randomnesscertifiedenergyconstrainedsemideviceindependent,li2026onchipquantumrandomnessamplification}. As noted in, e.g.,~\cite{brask2026quantumcorrelationsprepareandmeasurescenarios,Pauwels_2025}, multiple such approaches have been studied for prepare-and-measure (PM) protocols. Among them, 
\cite{Brask_2017,Wang_2019} consider PM set-ups, where the prepared states $\{\rho_x\}_{x\in \mathcal{X}}$ are impossible to  perfectly distinguish, i.e., the fidelity (see, e.g.,~\cite[Definition 3.7]{tomamichel2015quantum} for a definition) between any two prepared states satisfies $F(\rho_x,\rho_{x^\prime}) \geq \alpha_{xx^\prime}$ for strictly positive constants $\alpha_{xx^\prime}$.  A separate approach explored by the works~\cite{Li_2011,Lunghi_2015,PhysRevA.85.052308,Mironowicz_2016} is to base the protocol's security  on the assumption that the underlying Hilbert space of the prepared states is $d$-dimensional, where $d$ is known. Conversely,~\cite{Tavakoli_2020,Tavakoli_2022} show it is also possible to ensure the measurement output is weakly random without  directly constraining the prepared states or the underlying Hilbert space; rather they directly place constraints on the accessible information that is at the  eavesdropper's disposal.

One assumption that stands out for being particularly well 
physically motivated was first suggested in~\cite{Van_Himbeeck_2017,vanhimbeeck2019correlations}. While conceptually similar to set-ups that assume the prepared states are not fully distinguishable,~\cite{Van_Himbeeck_2017,vanhimbeeck2019correlations} impose energy constraints on the physical quantum systems. Concretely, they first assume that the underlying Hilbert space has a non-degenerate ground state and that there exists a gap between the two lowest energy eigenvalues. Photonic systems, e.g., can naturally satisfy these constraints~\cite{Van_Himbeeck_2017,vanhimbeeck2019correlations,Roch_i_Carceller_2025}. In other words, there exist quantum infrastructures that are naturally compatible with this semi-DI framework. 
Within this setting, the simplest semi-DI PM protocol consists of an untrusted source that, for each round, prepares one of two quantum states $\{\rho_x\}_{x \in \{0,1\}}$ (according to some predetermined distribution) that is subsequently measured by an untrusted measurement  $M$.\footnote{More generally, one can also consider a larger family of preparable states $\{\rho_{x}\}_{x \in \mathcal{X}}$ and multiple measurement settings $Y=y$ with corresponding measurement operators $M_y$.} After all rounds of the protocol are completed, the measurement outputs then make up the raw key. The only required assumption 
 is that each of the preparable states has an energy expectation value that is upper bounded by some predetermined value $\omega_x$. 
For low-energy states, i.e., $\omega_x \ll1$, the QRNG is then preparing and measuring quantum states that are close to a \textit{unique} ground state and therefore close to each other.

\begin{figure}[htb]
\centering
\includegraphics[width=0.45\textwidth]{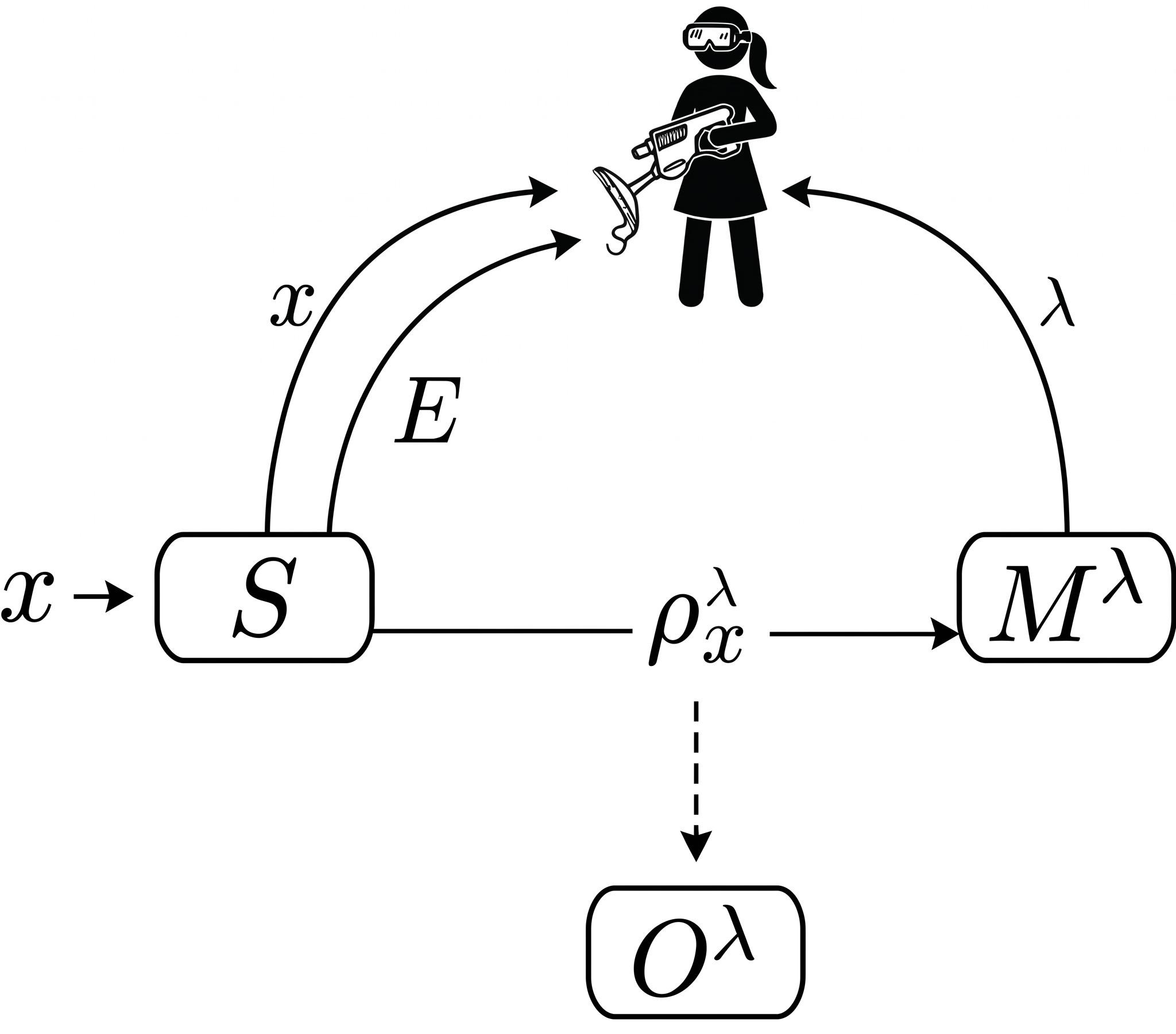}
\caption{Illustration of the general adversarial model, where the adversary has access to both the classical register $\Lambda$ storing the value $\lambda$ and the quantum register $E$. 
For every input $X=x$, the energy constraint is parametrized by a constant $\omega_x$, i.e., it holds that
$\sum_\lambda\operatorname{P}_{\Lambda} (\lambda) \operatorname{Tr}\left[ O^\lambda \rho^\lambda_x\right] \leq \omega_{x}$ where $O^\lambda$ is an energy operator.
\label{fig:our_model2}}
\end{figure}

The security of this set-up has primarily been studied against \textit{classical} adversaries~\cite{Van_Himbeeck_2017,vanhimbeeck2019correlations,Roch_i_Carceller_2025}, and experimental implementations have been shown to produce and certify randomness at reasonably high rates~\cite{PhysRevA.100.062338,10.1063/5.0011479,Tebyanian_2021,PhysRevApplied.15.034034}. However, this adversarial model is overly optimistic and one needs to also study eavesdroppers with access to quantum side-information.  In this work, 
we solve the question of lower bounding the single-round von Neumann entropy for a quantum adversary 
holding a general purification 
of the prepared state, as well as an additional classical register $\Lambda$. Concretely, for the simplest energy-constrained set-up with two source settings $X\in \{0,1\}$ and a single measurement setting, we provide direct lower bounds on the 
single-round von Neumann $H(A|E\Lambda,X=0)$. 

\begin{figure}[htbp]
  \centering
  \includegraphics[width=0.49\textwidth]{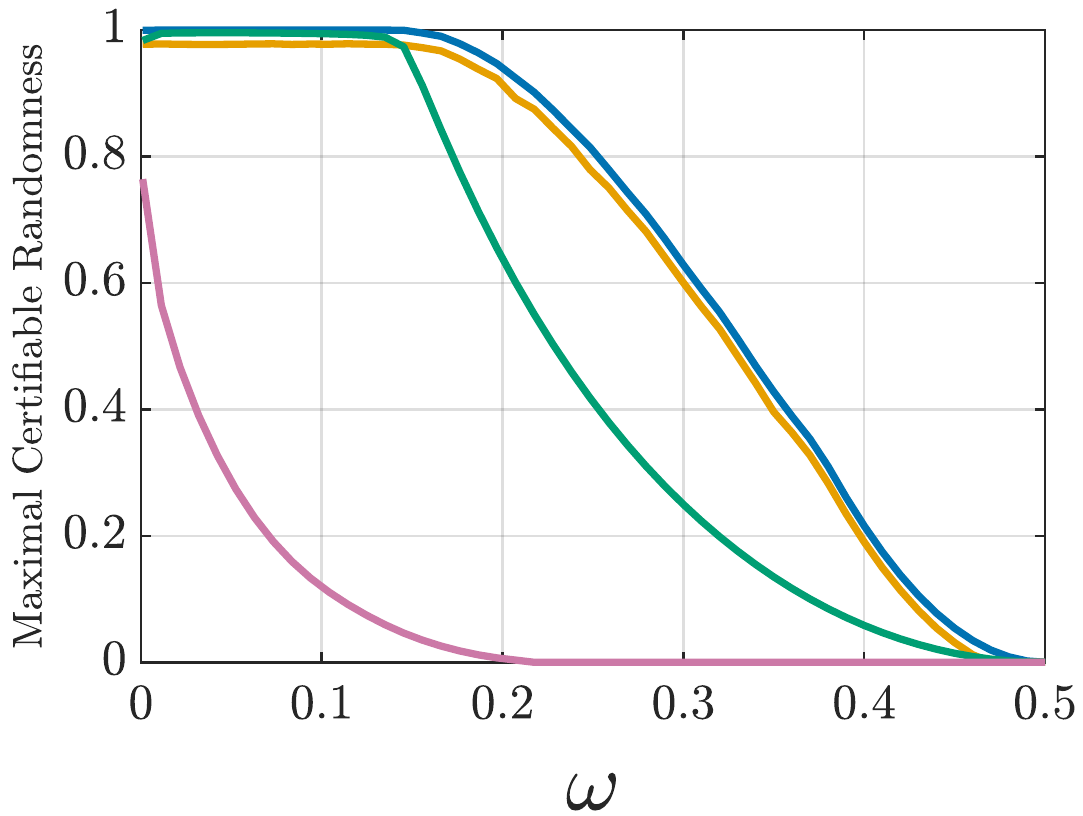} 
  
\caption{The maximal certifiable randomness as a function of $\omega=\omega_0=\omega_1$. The blue and orange lines denote the certifiable lower bounds on the von Neumann entropy  against  classical~\cite{vanhimbeeck2019correlations} and quantum adversaries (Theorem~\ref{Prop: SumReduction}), respectively. The green and purple lines denote lower boundds on the min-entropy against quantum~\cite{vanhimbeeck2019correlations} and entanglement-assisted adversaries~\cite{ranu2026randomnesscertifiedenergyconstrainedsemideviceindependent}. 
}
  \label{Intro:max-entropy-omega}
\end{figure}

The bounds derived in this work are in the form of a  hierarchy of semi-definite programs (SDPs), which provide increasingly tight lower bounds on the von Neumann entropy. To achieve these SDPs, we first apply Jordan's lemma~\cite{pironio2009deviceindependent,acin2007device} to reduce the dimension of the relevant source system to qubits. This step also allows us to bound the dimension of Eve's quantum register $E$. Subsequently, we
build upon the approach taken in~\cite{vanhimbeeck2019correlations,Hahn_2022} to only consider classical registers $\Lambda$
 of known and bounded dimension. The resulting bounds are very tight and robust against noise. We can also certify  non-trivial lower bounds on the conditional von Neumann against quantum attacks for all $\omega_x \lesssim 0.5$ that can often greatly surpass the min-entropy bounds of~\cite{vanhimbeeck2019correlations}; see Fig.~\ref{Intro:max-entropy-omega}.

\subsubsection*{Related Works}

 Recently, more works have explored energy-constrained PM set-up against eavesdroppers, Eve, that have access to additional quantum side-information. The works~\cite{DAvino:2025ymd,Roch_i_Carceller_2026,ranu2026randomnesscertifiedenergyconstrainedsemideviceindependent}, e.g., consider an entanglement-assisted version of this protocol. Here, Eve prepares an additional quantum state $\sigma_{SME}$, which is not constrained in any way. While Eve keeps $E$, the registers $SM$ are sent to the source and measurement device, respectively. For any input setting $X=x$, the source then acts on the register $S$ to prepare the energy-constrained state $\rho_{x}$. The subsequent measurement then acts jointly on $\rho_{x}$ and the register $M$ to produce the classical measurement output $A$. For this entanglement-assisted setting,~\cite{ranu2026randomnesscertifiedenergyconstrainedsemideviceindependent} have recently proven lower bounds on the single-round min-entropy $H_{\mathrm{min}}(A|E,X=0)$ which hold against all possible energy-constrained attacks. Moreover, these techniques 
are expected to extend towards achieving analogous lower bounds on the von Neumann and Petz-Rényi entropies, using the variational approaches of~\cite{Brown_2024, hahn2024,kossmann2025reliableentropyestimationobserved}. Such entanglement-assisted protocols are mathematically more general than the quantum adversaries we consider in this work. 
From our understanding, however, we do not attribute the improved entropy bounds found in this work  solely to the model, but to the difference in techniques and in the considered entropy.

Apart from entanglement-assisted set-ups,~\cite{senno2021}
 also consider the setting of this work, where Eve has access to both classical side-information and an additional quantum register that may be correlated or entangled with the prepared state $\rho_x$; in the most extreme case, this register holds a purification of the prepared state. Within this setting,~\cite{senno2021} provides a security proof against all adversaries whose energy-constrained attacks can be modeled via the entropy accumulation theorem (EAT).\footnote{There is a slight caveat in that~\cite{senno2021} uses the original entropy accumulation theorem from~\cite{DFR20}, which cannot adequately handle energy constraints. As such, the attacks considered in~\cite{senno2021} do not encompass all possible 
 quantum attacks.} For these attacks, the overall (globally) accumulated entropy over $n$ rounds of the protocol scales linearly with the worst-case von Neumann entropy that can be accumulated during a \textit{single} round of the protocol. However, since~\cite{senno2021}  bound this single-round quantity in terms of a lower bound on the min-entropy using the approach of~\cite{vanhimbeeck2019correlations}, this leaves open the question of how one can directly bound the von Neumann entropy instead.

Lastly, during the writing of this manuscript, we became aware of the recent experimental work~\cite{li2026onchipquantumrandomnessamplification}. Using a different approach,~\cite{li2026onchipquantumrandomnessamplification} independently derive lower bounds on the von Neumann entropy for both the quantum setting we consider in this work and
the entanglement-assisted setting, via the Brown-Fawzi-Fawzi method~\cite{Brown_2024}.   When directly comparing  our results to~\cite[Fig.~2 and Fig.~4]{li2026onchipquantumrandomnessamplification}, we appear to derive significantly tighter bounds on the von Neumann entropy. However,
we cannot provide a direct comparison, as~\cite{li2026onchipquantumrandomnessamplification} do not provide sufficient theoretical results or numerical material for us to replicate their results. 

\section{Background}

\subsection{Energy-Constraint Model \& Corresponding Quantum Correlations}

We consider the PM set-up originally proposed in~\cite{Van_Himbeeck_2017,vanhimbeeck2019correlations}. In their set-up, a source takes as input some classical value $x \in \{0,1\}$, and produces the corresponding quantum state $\rho_x$. Since we are considering a semi-DI setting, we assume that the underlying Hilbert space which contains $\{\rho_x\}_{x\in \{0,1\}}$ is of some \textit{unknown}, but finite, dimension. This state is then measured, and the output is denoted by $a \in \{-1,1\}$. We denote the measurement operator by $M$. Crucially, this measurement does not depend on the choice of $x$.

\begin{figure}[htb]
\centering
\includegraphics[width=0.45\textwidth]{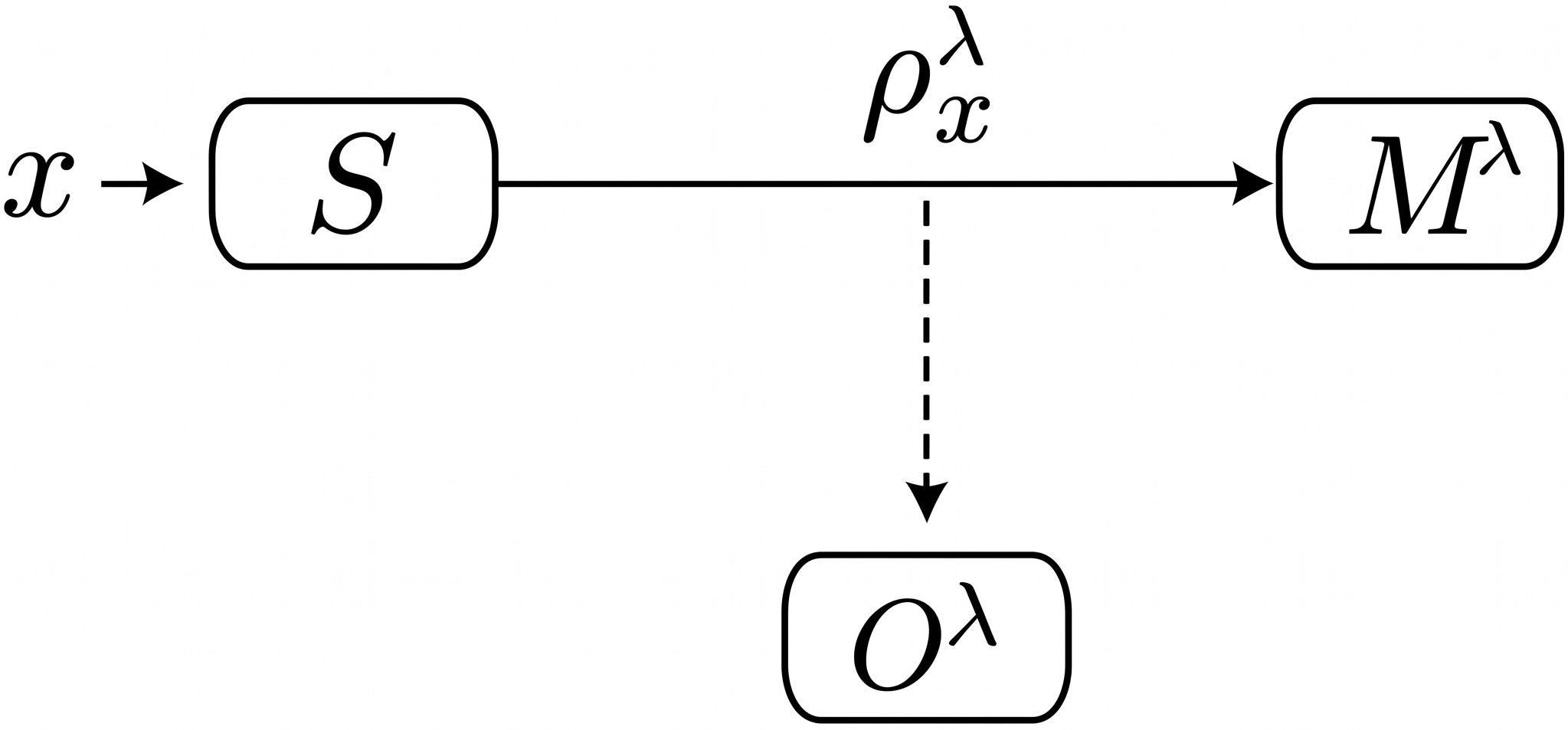}
\caption{The source prepares a quantum state $\rho_{x}^{\lambda}$ that may depend on the input $X=x$ and an additional classical register holding $\Lambda=\lambda$ with probability $\operatorname{P}_{\Lambda}(\lambda)$. The subsequent measurement $M^\lambda$ may depend on $\Lambda$, but not on $X$. For each choice $\Lambda =\lambda$, $O^\lambda$ denotes the corresponding energy operator. The states generated by the source satisfy the following energy constraint: For every input $X=x$, the energy constraint is parametrized by a constant $\omega_x$, i.e., it holds that
$\sum_\lambda\operatorname{P}_{\Lambda} (\lambda) \operatorname{Tr}\left[ O^\lambda \rho^\lambda_x\right] \leq \omega_{x}$.}
\label{fig:our_model}
\end{figure}

For the context of this work, apart from certain assumed energy constraints, the honest party, Alice, treats the source as well as the measurement device, as black boxes. Conversely, the eavesdropper, Eve, has full control over their design. Moreover, she has access to two additional forms of information. She is allowed to hold some classical register, $\Lambda$, which may be correlated to both the state that is produced by the source $S$, as well as the measurement device $M$. From now on, we therefore  denote the corresponding states and measurements via $\rho_x^\lambda$ and $M^\lambda$, respectively. 

 Moreover, for any value $\Lambda = \lambda$, Eve is additionally allowed to hold a purification of the state produced by the source, which is stored in a register $E$. This expands on the set-up in~\cite{Van_Himbeeck_2017,vanhimbeeck2019correlations}, which does not allow the adversary to hold any quantum side-information. Following~\cite{Van_Himbeeck_2017,vanhimbeeck2019correlations}, we introduce observables $O^\lambda$, which we shall henceforth refer to as energy operators, with the following properties:
\begin{itemize}
    \item The lowest eigenvalue of $O^\lambda$ is $0$. All other eigenvalues are greater or equal to $1$.
    \item The eigenspace associated with the eigenvalue $0$ is one dimensional, i.e., it is non-degenerate.
\end{itemize}
The only additional condition that we then impose on the system is a constraint on the average energy of the state emitted by the source, i.e., for every $\Lambda=\lambda$ there exists an energy operator $O^\lambda$ such that
\begin{align}\label{eq:energy-constraint}
\sum_\lambda\operatorname{P}_{\Lambda} (\lambda) \operatorname{Tr}\left[ O^\lambda \rho^\lambda_x\right] \leq \omega_{x} \; ,
\end{align}  
where $\operatorname{P}_{\Lambda} (\lambda)$ denotes a normalized probability distribution. The constraint in Eq.~\eqref{eq:energy-constraint} limits the possible correlations between the source and the measurement device. Intuitively, tighter energy bounds (i.e., smaller values of $\omega_{x}$) force Eve to choose strategies where $\rho_0^\lambda$ and $\rho_1^\lambda$ are closer to the vacuum state, and hence  closer to each other.
 More concretely, any attack by the adversary can  be described by a tuple $\left(\operatorname{P}_\Lambda , \{ \rho_{SE\Lambda,x}\}_{x}, \{ M^{\lambda} \}_\lambda , \{ O^{\lambda} \}_\lambda \right)$, which we define below.
\begin{definition}
A \textit{quantum strategy} is denoted by the tuple $\left(\operatorname{P}_\Lambda, \{ \rho_{SE\Lambda,x}\}_{x}, \{ M^{\lambda} \}_\lambda , \{ O^{\lambda} \}_\lambda\right)$. Here, $\operatorname{P}_\Lambda$  is a normalized probability distribution on the classical register $\Lambda$. Moreover, for any $x$ and $\Lambda=\lambda$, there exists a  normalized, bipartite quantum state $\rho_{SE,x}^{\lambda}$ between Alice and Eve such that 
\begin{align}
  \rho_{SE\Lambda,x}= \sum_{\lambda}\operatorname{P}_\Lambda(\lambda)\rho_{SE,x}^{\lambda} \otimes\ketbra{\lambda}{\lambda}\; .
\end{align}
 Lastly, for any $\Lambda=\lambda$, $M^{\lambda} $ and $O^{\lambda} $ are valid measurement and energy operators on $S$, respectively.
\end{definition}
For any input $x$ and $\Lambda = \lambda$, we denote Alice's  corresponding expected measurement outcome by
\begin{equation}    E^\lambda_x \coloneq \operatorname{Tr}\left [M^\lambda\rho^\lambda_x\right]\;, \quad \mathrm{where} \  \rho_{x}^{\lambda} \coloneq \text{Tr}_{E} \left[ \rho_{SE,x}^{\lambda} \right] \; .
\end{equation}
Moreover, the correlations characterizing the output of the device as seen from the point of view of a user which does not have access to the value of the hidden variable $\Lambda$ are given by $(E_0,E_1)$, where
\begin{equation}
E_x \coloneqq \sum_{\lambda} \operatorname{P}_{\Lambda}\left(\lambda\right) E_x^\lambda\;.
\end{equation}
Since tighter energy bounds  limit the quantum strategies available to Eve, it also reduces the amount of information she can potentially hold regarding the measurement outcome.   
The performance of the device is therefore characterized by the tuple $(\bm{E},\bm{\omega}) \coloneq (E_0,E_1,\omega_0,\omega_1)$. We refer to such a tuple  as a \emph{(quantum) behavior}.

\begin{definition}
A tuple $\left(E_0,E_1,\omega_0,\omega_1\right)$ specifying a possible pair of values for the correlations $\left(E_0,E_1\right)$ between the source and measurement device, and a possible pair of upper bounds $\left(\omega_0,\omega_1\right)$ on the averaged energies of the state emitted by the source is called a (quantum) behavior.    
\end{definition}

The set of behaviors that are realizable within quantum physics is denoted by $\mathcal{Q}$ and it has been characterized in~\cite{Van_Himbeeck_2017,vanhimbeeck2019correlations}.

\begin{definition}
A behavior $(E_0,E_1,\omega_0,\omega_1)$ belongs to the set $\mathcal{Q}$ if there exists a quantum strategy  $\left(\operatorname{P}_\Lambda, \{ \rho_{SE\Lambda,x}\}_{x}, \{ M^{\lambda} \}_\lambda , \{ O^{\lambda} \}_\lambda\right)$ such that 
\begin{align}
    &\sum_\lambda \operatorname{P}_\Lambda (\lambda) \operatorname{Tr}[\rho_x^\lambda M^\lambda]=E_x\;,\\
    &\sum_\lambda \operatorname{P}_\Lambda (\lambda) \operatorname{Tr}[\rho_x^\lambda O^\lambda]\leq \omega_x\;.\label{energy-constraint}
\end{align}
\end{definition}

\begin{lemma}[Quantum Set $\mathcal{Q}$ ~\cite{Van_Himbeeck_2017,vanhimbeeck2019correlations}] \label{Lem: SDPRep} A behavior $\left(E_0,E_1,\omega_0,\omega_1\right)$ belongs to the quantum set $\cal{Q}$ if and only if there exist $u,v,\eta_0,\eta_1\in\mathbb{R}$ satisfying $\eta_{0,1}\leq \omega_{0,1}$ such that the matrix   \begin{equation}
    \begin{aligned}
       \begin{pmatrix}1&u&E_0&2\eta_0-1\\
        u&1& E_1&2\eta_1-1\\E_0& E_1&1&v\\2\eta_0-1&2\eta_1-1&v&1
        \end{pmatrix}
    \end{aligned}
    \end{equation}
    is positive semi-definite. Moreover, if the behavior satisfies the trivial constraints $-1\leq E_x \leq 1$ 
and $\omega_x \geq 0$ as well as the condition $\omega_0+\omega_1\leq 1$,  then it belongs to the quantum set ${\cal Q}$ if
and only if
\begin{multline}
|\arcsin{E_0}-\arcsin{E_1}|\leq \\2\left(\arcsin\sqrt\omega_0+\arcsin\sqrt\omega_1\right)
\;.\label{Q-criterion-arcsin}
\end{multline}
\end{lemma}

\subsection{Single-Round von Neumann Entropy Contribution}

In this work, we are interested in the single-round quantity $H(A|E\Lambda ,X=0)_{\rho}$, where $A$ corresponds to Alice's measurement outcome that is attained when measuring $\rho_{SE\Lambda,x}$ for $x=0$,  and $H$ denotes the conditional von Neumann entropy which is defined as follows~\cite[Definition~5.1]{tomamichel2015quantum}. We use $\log$ to denote the base-2 logarithm throughout this work.
\begin{definition}
    Given any trace-normalized, bipartite quantum state $\rho_{AE}\in S_{=}\left(AE\right) $, the corresponding \textit{conditional von Neumann entropy} is given by
\begin{multline}
    H(A|E)_{\rho} \coloneq  H(AE)_{\rho}- H(E)_{\rho}\;, \\\quad \mathrm{where} \quad H(A)_{\rho}\coloneq -\operatorname{Tr}[\rho_{A}\log\left(\rho_{A}\right)] \; .
\end{multline}
Moreover, if the state is clear from context, then the subscript may be omitted.
\end{definition}

The conditional von Neumann entropy is of particular interest, as it typically quantifies the (asymptotic) rate at which entropy accumulates in typical quantum cryptography protocols. This relation is made precise via \textit{entropy accumulation theorems}~\cite{DFR20,Arnon-Friedman:2018aa,DF19,Liu_2021,pgrnmz9j,arqand2025marginalconstrainedentropyaccumulationtheorem,kamin2026unconditionalsecuritydiscretemodulatedcvqkd} or the closely related \textit{quantum probability estimation} framework~\cite{ZFK20,Knill_2020,knill2023quantumprobabilityestimationrandomness}.

We assume that the classical local hidden variable $\Lambda$ and the input $X$ satisfy \textit{measurement independence}~\cite{scarani2013deviceindependent}, i.e., 
\begin{align}
\operatorname{P}_{\Lambda X}\left(\lambda, x\right) = \operatorname{P}_{\Lambda}\left(\lambda\right)\cdot \operatorname{P}_{X}\left(x\right) \; .
\end{align}
 In this case, the conditional von Neumann entropy can be expressed as
\begin{multline}\label{objective-function}    
H(A|E\Lambda X)_{\rho} = \\\sum_{\lambda,x}\operatorname{P}_{\Lambda}\left(\lambda\right)\cdot \operatorname{P}_{X}\left(x\right) \cdot H(A|E,\Lambda=\lambda,X=x)_{\rho}
\;,
\end{multline}
where $H(A|E,\Lambda=\lambda,X=x)_{\rho}$ denotes the conditional von Neumann entropy evaluated on the post-measurement state $\rho^{\lambda}_{AE,x} = \left(M^\lambda\otimes \mathbb{I}_{E}\right) \left(\rho^{\lambda}_{SE,x}\right)$. 
We highlight that the measurement operator does not depend on  $X$, and the resulting problem of finding Eve's optimal strategy is a coupled optimization problem.

We are primarily interested 
in the limit $\operatorname{P}_{X}\left(0\right) \to 1$. This special case is analogous to DI cryptography set-ups where there is only a  single key-generating measurement; see, e.g.~\cite{pironio2009deviceindependent,acin2007device,hahn2025analyticrenyientropybounds}. In this case, the objective function simplifies to 
\begin{multline}\label{objective-function-X=0} 
    H(A|E \Lambda, X=0)_{\rho}= \\ \sum_{\lambda} \operatorname{P}_{\Lambda}\left(\lambda\right)H(A|E,\Lambda=\lambda, X=0)_{\rho} \; .
\end{multline}
During the following section, we aim to lower bound $H(A|E\Lambda,X=0)_{\rho}$ in this semi-DI setting. Concretely, we derive a hierarchy of SDPs which provide improved lower bounds on this conditional entropy from Eq.~\eqref{objective-function-X=0}.
\section{Semi Device-Independent Bounds on \texorpdfstring{$ H(A|E\Lambda, X=0)$}{Conditional von Neumann Entropy}}
Our main result in this section is a hierarchy of SDPs that can (increasingly tightly) lower bound $ H(A|E\Lambda, X=0)$. To this extent, we first provide analytic lower bounds for ${H(A|E,\Lambda=\lambda, X=0)}$ in Section~\ref{Sec: TightExpForFixedLambda}. We subsequently show in Section~\ref{bound-T-section} how to linearize this analytic expression in order to attain a \textit{lower} bound on $ H(A|E\Lambda, X=0)_{\rho}$ that is solvable via an SDP. In particular, the technique we utilize to linearize ${H(A|E,\Lambda=\lambda, X=0)}$ results in a hierarchy of SDPs that yields increasingly tighter lower bounds on the desired entropy.

\subsection{Entropy Bounds for Fixed \texorpdfstring{$\Lambda = \lambda$}{Fixed $\lambda$}}
\label{Sec: TightExpForFixedLambda}

In order to derive an analytic lower bound for ${H(A|E,\Lambda=\lambda, X=0)}$, we first employ Jordan's lemma  in Section~\ref{Sec: QubitReductions} and show that there exists a qubit reduction, i.e., it suffices to consider states where $SE$ are each qubit registers. For such two-qubit states, we then derive an expression for the post-measurement conditional entropy ${H(A|E,\Lambda=\lambda, X=0)}$ in Section~\ref{Sec: AnalyticLowBoundsFixedLambda}.

\subsubsection{Qubit Reductions} \label{Sec: QubitReductions}

\begin{definition}\label{definition: simplified energy operator}
Let $|0\rangle$ be the normalized eigenvector corresponding to the lowest eigenvalue of an energy operator $O$. We define the simplified energy operator as  $\widetilde{O}=\mathbb{I}-|0\rangle \langle 0|$.
\end{definition}
If $\left\{ \left|j\right\rangle \right\}_{j\in J}$ is a basis of normalized \textit{excited} eigenstates of $O$ such that  the energy operator $O$ is of the form
\begin{align}
    O = \sum_{j\in J} \alpha_j \cdot |j\rangle \langle j|\;,\;\;\; \alpha_j\geq 1\;\forall j \in J\;,
\end{align}
then the simplified energy operator $\widetilde{O}$ is simply given by
\begin{align}
 \widetilde{O} = \sum_{j\in J} |j\rangle \langle j|
\;.   
\end{align}
In this section we consider fixed $\Lambda = \lambda$ and aim to derive a lower bound on the conditional von-Neumann entropy ${H(A|E,\Lambda=\lambda, X=0)}$. Moreover, we relax the energy constraint in Eq.~\eqref{energy-constraint} by the weaker constraint
\begin{equation}\label{weak-energy-constraint}\sum_\lambda \operatorname{P}_{\Lambda}\left(\lambda\right)\operatorname{Tr}[\widetilde{O}^\lambda \rho^\lambda_x ] \leq \omega_{x}\;.
\end{equation} 
The bound that we subsequently derive is guaranteed to apply for the original set-up with the constraint Eq.~\eqref{energy-constraint}, as $\operatorname{Tr}[ \widetilde O^\lambda \rho ]\leq  \operatorname{Tr}[ O^\lambda \rho]$
for any state $\rho$.  Since the energy operator $ \widetilde O$ is a valid energy operator, the derived bounds are the best one may hope for.

\begin{rem} \label{Rem: ProjMeasReduc}
Following~\cite{vanhimbeeck2019correlations}, for the sake of lower bounding the conditional von Neumann entropy, one can ignore strategies in which $M^\lambda=\pm \mathbb{I}$ for some $\lambda$. Such deterministic attacks yield $E^\lambda_x=\pm 1$, and contribute zero to the conditional entropy. The same effect can be achieved using ${M^\lambda=\pm\left(\ketbra{0}{0}-\ketbra{j}{j}\right)}$ and choosing the state $\rho_x^\lambda=\ketbra{0}{0}$, where $\ket{0}$ and $\ket{j}$ are $O^\lambda$ eigenstates with eigenvalues $0$ and $\alpha_j \geq 1$, respectively. This attack satisfies $\text{Tr}\left[O^\lambda\rho_x^\lambda \right]=0$ and is therefore compatible with the behavior of the original attack.  
We will henceforth assume that $M^\lambda \neq\pm \mathbb{I}$ for all $\lambda$.

Moreover, without loss of generality, we assume that $M^\lambda$ is a projective measurement. This follows from the fact that, for binary outputs, any feasible POVM on finite dimensions can be expressed in terms of a convex mixture of projective measurements~\cite[Lemma 1]{Wright_2019}. Applying this convex mixture instead and giving Eve the classical information as to which projective measurement was applied leads to an attack that cannot be weaker than the original POVM strategy due to data processing~\cite[Corollary 5.1]{tomamichel2015quantum}.
\end{rem}

\begin{lemma}\label{claim: decomposition of the operators to 3 Hilbert spaces}
Let $M$ and $\widetilde O$ be Hermitian operators on a Hilbert space $\cal{H}$ such that the eigenvalues of $M$ are $+1$ and $-1$, each with degeneracy at least one, and the eigenvalues of $\widetilde O$ are $0$ with degeneracy one, and $1$ with degeneracy at least one.   
There exists a decomposition of $\mathcal{H}$ into a direct sum of three orthogonal subspaces: $\mathcal{H}=\mathcal{H}_{g}\oplus\mathcal{H}_{1}\oplus\mathcal{H}_{-1}$ such that $\mathcal{H}_{g}$,$\mathcal{H}_{1}$ and $\mathcal{H}_{-1}$ are conserved spaces of both $M$ and $\widetilde{O}$, and such that the restrictions of $M$ and $\widetilde{O}$ to these subspaces obey
\begin{align}
    &M_1=\mathbb{I}_{1}
    \;,
    \;\;\;\;\;\;\;\;
    \;\widetilde O_{1}=\mathbb{I}_{1}
    \;,\\
    &M_{-1}=-\mathbb{I}_{-1}
    \;,\;\;\widetilde O_{-1}=\mathbb{I}_{-1}
    \;.
\end{align}
Moreover, $\mathcal{H}_{g}$ is a two-dimensional space and there exists a basis for $\mathcal{H}_{g}$ in which 
\begin{equation} \label{Eq: MeasAndEnOperatorAngles}
    M_g = \sin\theta \sigma_1 + \cos\theta \sigma_3\;,\;\;\; O_g = \frac{\mathbb{I}_g - \sigma_3}{2}\;,
\end{equation}
for some $\theta\in[-\pi,\pi]$. Here, we used $M_i$ and $O_i$ for the restrictions of $M$ and $\widetilde O$ to $\mathcal{H}_i$ and $\mathbb{I}_i$ for the identity operator in $\mathcal{H}_{i}$,  $i\in\{1,-1,g\}$. Moreover, $\sigma_{1}$,$\sigma_{2}$ and $\sigma_{3}$ are Pauli matrices acting in $\mathcal{H}_g$. 
\end{lemma}

\begin{proof}
According to Jordan's lemma (see, e.g.,~\cite{pironio2009deviceindependent}), there exists a decomposition of $\mathcal{H}$ into a direct sum of orthogonal subspaces $\mathcal{H}=\bigoplus \widehat{\mathcal H}_{\beta}$ such that the $\widehat{\mathcal H}_{\beta}$'s are conserved under both $M$ and $\widetilde O$, and $\dim\widehat{\mathcal H}_{\beta}\leq 2$. 

The fact that the ground state $\ket{0}$ of $\widetilde O$ is non-degenerate implies that there exists $\beta^*$ such that ${|0\rangle\in \widehat{\mathcal H}_{\beta^*}}$. Indeed, if $|0\rangle=\sum_\beta a_\beta  |\chi_\beta\rangle$ where $|\chi_\beta\rangle\in\widehat{\mathcal H}_{\beta}$, then $0=\sum_\beta a_\beta  \widetilde O|\chi_\beta\rangle$ and thus ${a_\beta  \widetilde O|\chi_\beta\rangle=0}$ for all $\beta$. The uniqueness of the ground state implies that $a_\beta\neq 0$ exactly for one $
\beta$, which we denote by $\beta^*$.  

We begin with the case in which $\dim\mathcal{H}_{\beta^*}=1$,  i.e., 
$\mathcal{H}_{\beta^*}=\text{span}\left(|0\rangle\right)$. Thus, $|0\rangle$ is an eigenstate of $M$ with  eigenvalue $1$ or $-1$. We pick another eigenstate of $M$ with the opposite eigenvalue and define $\mathcal{H}_g$ to be the space spanned by the two states.  In this eigenstate basis $M_g=\pm\sigma_3$ and $O_g = \frac{\mathbb{I}_g - \sigma_3}{2}$. 

If  $\dim\mathcal H_{\beta^*}=2$, we define $\mathcal{H}_g=\mathcal{H}_{\beta^*}$. Next, we diagonalize $O_g$ to achieve $O_g=\frac{\mathbb{I}_g - \sigma_3}{2}$.   The eigenvalues of $M_g$ are $1$ and $-1$.\footnote{If $M_g$ was proportional to the identity operator, $\mathcal{H}_g$ would further decompose into two orthogonal conserved subspaces of dimension one and we would be back in the scenario in which $\dim
\mathcal H_{\beta^*}=1$. We thus assume that $M_g$ has both $1$ and $-1$ as eigenvalues.
}  Thus, $M_g$ can be written as $M_g=a_1\sigma_1+a_2\sigma_2+a_3\sigma_3$,  with real coefficients $a_i$.  Using a rotation around the z-axis, $\exp\left({\frac{i}{2}\arg(a_1+ia_2)\sigma_3}\right)$, we then switch to a basis in which $M_g = \sin\theta \sigma_1 + \cos\theta \sigma_3$  while keeping $O_g = \frac{\mathbb{I}_g - \sigma_3}{2}$ intact.

Finally, we define $\mathcal{H}_1$ and $\mathcal{H}_{-1}$ as the eigenspaces of $\mathcal{H}_g^\bot$ corresponding to eigenvalues $1$ and $-1$ of $M$, respectively. 
\end{proof}

\begin{rem}\label{remark:focusing on Hg}
     When determining Eve's optimal strategy, it is sufficient to focus solely on the $\mathcal{H}_g$ subspace. This is because employing deterministic strategies in the other subspaces, $\mathcal{H}_1$ and $\mathcal{H}_{-1}$, does not provide Eve with any advantage in terms of extracting information.  To illustrate this, let us consider the following argument.

 For any $\Lambda=\lambda$, let $\Pi_g $, $\Pi_1$, and $\Pi_{-1}$ denote the projectors onto the subspaces $\mathcal{H}_g$, $\mathcal{H}_1$, $\mathcal{H}_{-1}$, respectively. Following the argument of, e.g.,~\cite{pironio2009deviceindependent}, it suffices to consider states that are invariant under these projectors, i.e., that are block-diagonal on these subspaces. This is because non block-diagonal entries do not contribute to $\operatorname{Tr}
    \left[O \rho\right]$, $\operatorname{Tr}
    \left[M \rho\right]$, or the post-measurement conditional entropy $ H(A|E,\Lambda=\lambda, X=0)$. However, for any pair of states $\{\rho^{\lambda}_{SE,x}\}_x$ of the form
    \begin{multline} \label{Eq: GenSingLamAtt}
        \rho^{\lambda}_{SE,x} = p_{g,x} \rho_{g,x}  + p_{1,x} \rho_{1,x}  + p_{-1,x} \rho_{-1,x} \; , \\  \quad  \mathrm{where} \ \operatorname{Tr}_{E}[\rho_{y,x}]\in \mathcal{H}_y \ \mathrm{for}\  y \in \{g,+1,-1\} \; ,
    \end{multline}
    Eve can construct the following alternate attack. She can first create an additional classical random variable $\Lambda^\prime$ which stores $\Lambda^\prime= (y_0,y_1)$ 
with probability $p_{y_0,0}\cdot p_{y_1,1}$, where $y_x \in \{g,+1,-1\}$. Here, we note that if Eve were to keep $\Lambda^\prime$ and distribute the pair of states $\{\rho_{y_x,x}\}_x$ instead, then this attack generates (on average) the desired energy and measurement expectation values $(E_0^{\lambda},E_1^{\lambda},\omega_0^{\lambda},\omega_1^{\lambda})$ without reducing Eve's side-information (due to data processing~\cite[Corollary 5.1]{tomamichel2015quantum}). However, instead of sending the states $\rho_{\pm 1,x}$ (which produce deterministic outcomes) she can also instead send $\ketbra{m^\pm}{m^{\pm}}_{S} \otimes \ketbra{0}{0}_E$, where $\ket{m^+},\ket{m^-}\in\mathcal{H}_g$ are eigenstates of $M_g^\lambda$ with eigenvalues $+1$ and $-1$. These states also produce deterministic outcomes and Eve can fully determine the outcome via $\Lambda^\prime$. Indeed, it is easy to verify that the resulting attack again achieves the quantum behavior  $(E_0^{\lambda},E_1^{\lambda},\omega_0^{\lambda},\omega_1^{\lambda})$. This resulting attack gives Eve at least the same level of information as the attack from Eq.~\eqref{Eq: GenSingLamAtt}, while using only the subspace $\mathcal{H}_g$. As such, one may restrict Eve to $\mathcal{H}_g$ without weakening the optimal attack.   
\end{rem}

We conclude this section by noting that since, for any $\Lambda = \lambda$, it is sufficient to consider the simplified 
 energy operator $\widetilde O^\lambda$, Lemma~\ref{claim: decomposition of the operators to 3 Hilbert spaces} can be applied.  Moreover, a consequence of Remark~\ref{remark:focusing on Hg} and the data processing inequality is that it suffices to consider $\rho_{SE,x}^\lambda$ where both $S$ and $E$ are qubit registers.
 For any states $\rho_{SE,x}^\lambda$ and operators $M^\lambda,O^\lambda$ such that 
 \begin{align}
    \operatorname{Tr} \left[ M^\lambda \rho_{SE,x}^\lambda\right] = E_x^\lambda \quad \mathrm{and} \quad \operatorname{Tr} \left[ O^\lambda \rho_{SE,x}^\lambda\right] \leq \omega_x^\lambda \; ,
 \end{align}
it namely follows from the Remark~\ref{remark:focusing on Hg} that there exists: a distribution $\operatorname{P}_{\Lambda^\prime}$, a family of qubit-qubit states $\{\rho_{SE,x}^{\lambda^\prime}\}_{\lambda^\prime}$, single-qubit observables $\{M^{\lambda^\prime}\}_{\lambda^\prime}$, and single-qubit energy operators $\{O^{\lambda^\prime}\}_{\lambda^\prime}$ 
such that 
 \begin{align}
    \sum_{\lambda^\prime} \operatorname{P}_{\Lambda^\prime}(\lambda^\prime)\operatorname{Tr} \left[ M^{\lambda^\prime} \rho_x^{\lambda^\prime}\right] &= E_x^\lambda \; , \\  \sum_{\lambda^\prime}\operatorname{P}_{\Lambda^\prime}(\lambda^\prime)\operatorname{Tr} \left[ O^{\lambda^\prime} \rho_x^{\lambda^\prime}\right] &\leq \omega_x^{\lambda} \; ,  \\ H(A|E, \Lambda = \lambda, X=0) &\geq H(A|E\Lambda^\prime, X=0) \;  .
 \end{align}
 Moreover, for all $\Lambda^\prime=\lambda^\prime$, $M^{\lambda^\prime}$ is a measurement on a qubit space with   non-degenerate eigenvalues $\pm 1$. Similarly, $O^{\lambda^\prime}$ is a rank-one energy operator on the same qubit space with non-degenerate eigenvalues $0,1$.

\subsubsection{Lower Bounds for Fixed \texorpdfstring{$\Lambda = \lambda$}{Fixed $\Lambda = \lambda$}} \label{Sec: AnalyticLowBoundsFixedLambda}
Without loss of generality, we assume that the initial states $\rho_{SE,x}^{\lambda}$ are (pure) states on \textit{qubit} registers $SE$ and that $M^\lambda,O^\lambda$ are rank-one measurements on $S$. Our main result in this section is an 
 lower bound on $ {H(A|E,\Lambda=\lambda, X=0)}$, where $A$ denotes Alice's measurement outcomes after applying $M$.  For the resulting expression provided in Theorem~\ref{theorem: optimization problem for the general model} below, we require the following SDP for any given quantum behavior $(E_0,E_1,\omega_0,\omega_1)$:
\begin{equation}\label{Eq: QstarExp}
    \begin{aligned}
        & q_\star \coloneqq && \min_{t,z_0,z_1} q\\
        & && \text{subject to} \\
         & &&  \begin{pmatrix}
  q & E_0 & z_0 \\
  E_0 & 1 & t \\
  z_0 & t & 1
\end{pmatrix} \succeq 0 \; , \quad \begin{pmatrix}
  1 & E_1 & z_1 \\
  E_1 & 1 & t \\
  z_1 & t & 1
\end{pmatrix} \succeq 0\; , \\
        & &&  1\geq z_x \geq 1-2\omega_x \ \forall x\in \{0,1\} \; , \quad 0\leq q \leq 1 \;.
    \end{aligned}
    \end{equation}

\begin{theorem}\label{theorem: optimization problem for the general model}
    Let $\{\rho_{SE,x}^{\lambda}\}_{x\in\{0,1\}}$ be any set of two two-qubit states. Moreover, let $M^\lambda$ and $\widetilde O^\lambda=\mathbb{I}-\ketbra{0}{0}$ denote a rank-one measurement and (simplified) energy operator on $S$,\footnote{Note that one may, without loss of generality, assume that one of the operators is diagonal in the computational basis. Here, we choose to express $\widetilde O^\lambda$ to be diagonal in this basis.} respectively, such that the resulting behavior is given by $\left(E_0^\lambda,E_1^\lambda,\omega_0^\lambda,\omega_1^\lambda\right)\in \cal{Q}$. Then
\begin{multline}\label{general-bound-T}
        H\left(A|E,\Lambda=\lambda, X=0\right) \geq \\  h_2\left(\frac{1-E_0^\lambda}{2}\right)-h_2\left(\frac{1-\sqrt{q_\star}}{2}\right)\;,
    \end{multline}  
    where $q_\star$ is given by Eq.~\eqref{Eq: QstarExp}. Conversely, this bound is optimal for any behavior ${\left(E_0^\lambda,E_1^\lambda,\omega_0^\lambda,\omega_1^\lambda\right)\in \cal{Q}}$ that is realizable in two  dimensions with a family $\{\rho_{S,x}^{\lambda}\}_{x\in\{0,1\}}$  of qubit states and corresponding projective operators $M^\lambda$, $\widetilde O^\lambda$. That is, there then exist two two-qubit states on $SE$ and corresponding operators that saturate this bound.
\end{theorem}

\begin{proof}
     Without loss of generality, we assume that  $\rho^{\lambda}_{SE,x}$ are pure states. Moreover, following Lemma~\ref{claim: decomposition of the operators to 3 Hilbert spaces}, we will choose the basis such that
     \begin{align}
         M^\lambda = \sqrt{1-t^2} \sigma_1 + t \sigma_3\;,\quad  \widetilde O^{\lambda} = \frac{\mathbb{I} - \sigma_3}{2}\;,
     \end{align}
where $|t|\leq 1$. We denote the corresponding measurement and energy expectation values by  $E_x^\lambda\equiv\operatorname{Tr}\left [M^\lambda\rho^\lambda_x\right] $ and $\widetilde\omega_x^\lambda\equiv\operatorname{Tr}\left [\widetilde O\rho^\lambda_x\right] $, respectively, and highlight that  $\widetilde\omega_x^\lambda \leq \omega_x^\lambda$ holds by definition. Let us now consider the conditional von Neumann entropy 
\begin{multline}
    H(A|E,\Lambda=\lambda ,X=0) =  \\ H(AE|\Lambda=\lambda,X=0) - H(E|\Lambda=\lambda,X=0) \;.
\end{multline}
  Here, $H(AE|\Lambda=\lambda,X=0)$  is easily computed in the measurement eigenbasis $\{ \ket{+}_{S}, \ket{-}_{S}\}$ of $M^\lambda$  with corresponding eigenvalues $\{ +1,-1\}$. We can express the shared pre-measurement state $\rho^{\lambda}_{SE,0}= \ketbra{\psi_{SE}^{\lambda}}{\psi_{SE}^{\lambda}}$ as 
  \begin{multline}
    \ket{\psi_{SE}^{\lambda}} = \\ \sqrt{1-f_{\lambda}^2}\,\ket{+}_{S}\otimes \ket{\psi_1}^{\lambda}_{E} + f_{\lambda}\,\ket{-}_{S}\otimes \ket{\psi_2}^{\lambda}_{E} \; ,
\end{multline}
where $f_{\lambda}\in [0,1]$ such that $E_0^\lambda =1-2f_{\lambda}^2$. Moreover, $ \ket{\psi_1}_{E}^{\lambda}, \ket{\psi_2}^{\lambda}_{E}$ are  normalized vectors in Eve's Hilbert space.\footnote{Any complex pre-factor can be absorbed into the definitions of $\{ \ket{\psi_1}_{E}^{\lambda}, \ket{\psi_2}^{\lambda}_{E}\}$.}
  Evaluating this quantity on the post-measurement state
\begin{multline}
    \rho_{AE,x}^{\lambda} = \left(1-f_{x,\lambda}^2\right) \ketbra{+}{+}_{A} \otimes \ketbra{\psi_1}{\psi_1}^{\lambda}_{E,x} \\+ f_{x,\lambda}^2 \ketbra{-}{-}_{A} \otimes \ketbra{\psi_2}{\psi_2}^{\lambda}_{E,x} \; ,
\end{multline}
one finds that
\begin{multline}\label{H(AE)}
  H\left(AE|\Lambda=\lambda,X=0\right)=\\
  h_{2}\left(f_{0,\lambda}^2\right)= h_{2}\left(\frac{1-E_0^\lambda}{2}\right)\; .
\end{multline}   
     To derive an expression for $H(E|\Lambda=\lambda,X=0)$ on the post-measurement state, recall that this is equal to evaluating  $H(E|\Lambda=\lambda,X=0)$ on the pre-measurement state $\rho_{SE,0}^{\lambda}$. Moreover, since we are assuming this state is pure, this is simply equal to $H(S|\Lambda=\lambda,X=0)$.\footnote{This can readily be seen using the Schmidt decomposition; see, e.g.,~\cite[Theorem 2.2]{khatri2024principlesquantumcommunicationtheory}} The reduced states on $S$ can be expressed as 
     \begin{align}
         \rho_{S,x}^{\lambda} =\frac{1}{2}\left(\mathbb{I}+{\bf n}_x\cdot{\boldsymbol{\sigma}}\right)\; ,
     \end{align}
    where $\boldsymbol{\sigma} = (\sigma_1,\sigma_2,\sigma_3)^{\operatorname{T}}$ denotes the vector containing the Pauli matrices and ${\bf n}_x= (n_{x,1},n_{x,2},n_{x,3})^{\operatorname{T}}$ is a sub-normalized three-dimensional vector. Such states have eigenvalues $\frac{1 \pm|{\bf n}_x|}{2}$ (see, e.g.,~\cite[Eq.~(2.94)]{Barnett2009QuantumInformation}) and the von Neumann entropy for $X=0$ can therefore simply be expressed as
    \begin{align} \label{Eq: EntNormN}
       H(S|\Lambda=\lambda,X=0)  = h_2\left(\frac{1-|{\bf n}_0|}{2}\right) \; .
    \end{align}
To achieve an upper bound on this quantity, we now aim to find $ \inf|{\bf n}_0|$, i.e., the smallest value that $|{\bf n}_0|$ may take as a function of the underlying quantum behavior. Concretely, the energy expectation value constraints are equivalent to  the inequalities $n_{x,3}  \geq 1-2\omega_x^\lambda$. The measurement expectation value constraints, on the other hand, can alternatively be expressed as $
    E_x^\lambda = \sqrt{1-t^2}n_{x,1} +tn_{x,3} $.
Note that we may without loss of generality set $n_{x,2}=0$, as the mapping $\bold{n}_x \mapsto (n_{x,1},0,n_{x,3})^{\operatorname{T}}$ maps sub-normalized vectors to sub-normalized vectors while leaving $(E_0^\lambda,E_1^\lambda, \widetilde{\omega}_0^\lambda,\widetilde{\omega}_1^\lambda)$ invariant.
Setting $y_x \coloneq n_{x,1}$ and $z_x \coloneq n_{x,3}$, the smallest feasible value for $|\bold{n}_0|^2$ is simply lower bounded by the following optimization problem:

\begin{equation}\label{Eq: MinNormN}
    \begin{aligned}
         & \min_{t,\{y_x,z_x\}} && y_0^2  + z_0^2 \\
        & \text{subject to} \\
        &  1\geq z_x \geq 1-2\omega_x &&\ \forall x\in \{0,1\} \; , \quad -1 \leq t \leq 1 \;, \\
         &  E_x^\lambda = \sqrt{1-t^2}y_x +tz_x && \ \forall x\in \{0,1\}\;  , \\
         & y_x^2  + z_x^2  \leq 1 &&\ \forall x\in \{0,1\}\; .
    \end{aligned}
    \end{equation}
   The feasible set is not empty  since, by assumption, there exists at least one family of single-qubit states $\{\rho_{S,x}^{\lambda}\}_{x\in\{0,1\}}$  and corresponding operators $M^\lambda$, $\widetilde O^\lambda=\mathbb{I}-\ketbra{0}{0}$ that achieve the desired quantum behavior $\left(E_0^\lambda,E_1^\lambda,\omega_0^\lambda,\omega_1^\lambda\right)\in \cal{Q}$. Moreover, the minimum is achieved since the objective function is continuous and the feasible set is closed and bounded. 
In particular, the last two constraints ensures that we are only optimizing over all $t,y_x,z_x$ for which there actually exists a feasible state (i.e., $|{\bf n}_x| \leq 1$) $\rho_{S,x}^{\lambda}$ that is compatible with the measurement expectation value $E_x$.
In Appendix~\ref{App: OptimalValueQ}, we show that Eq.~\eqref{Eq: MinNormN} is equivalent to
\begin{equation} \label{Eq: OptTildeQ}
    \begin{aligned}
        & q_\star \coloneqq && \min_{q,t,z_0,z_1} q\\
        & && \text{subject to} \\
         & &&  \begin{pmatrix}
  q & E_0^{\lambda} & z_0 \\
  E_0^{\lambda} & 1 & t \\
  z_0 & t & 1
\end{pmatrix} \succeq 0 \; , \quad \begin{pmatrix}
  1 & E_1^{\lambda} & z_1 \\
  E_1^{\lambda} & 1 & t \\
  z_1 & t & 1
\end{pmatrix} \succeq 0\; , \\
        & &&  1\geq z_x \geq 1-2\omega_x \ \forall x\in \{0,1\} \; , \quad 0\leq q \leq 1 \;.
    \end{aligned}
    \end{equation}
Since this is a valid lower bound on $|\bold{n}_0|^2$, we get the following bound on the desired entropy via Eq.~\eqref{Eq: EntNormN}
\begin{align} \label{Eq: OptHE}
       H(S|\Lambda=\lambda,X=0)  \leq h_2\left(\frac{1-\sqrt{q_\star}}{2}\right)\; .
    \end{align}
   Combining  Eq.~\eqref{Eq: OptHE} with Eq.~\eqref{H(AE)} yields the desired bound on $H(A|E,\Lambda=\lambda ,X=0)$. 

We now prove the tightness of this bound for all quantum behaviors $\left(E_0^\lambda,E_1^\lambda,\omega_0^\lambda,\omega_1^\lambda\right)\in \cal{Q}$ that are achievable with qubit states $\{\rho_{S,x}^{\lambda}\}_{x\in\{0,1\}}$  and corresponding projective operators $M^\lambda$, $\widetilde O^\lambda$. Since $\left(E_0^\lambda,E_1^\lambda,\omega_0^\lambda,\omega_1^\lambda\right)\in \cal{Q}$ is achievable in two dimensions, the feasible set in Eq.~\eqref{Eq: OptTildeQ} has to be non-empty by construction and therefore $0\leq q_{\star} \leq 1$ naturally holds.  For this optimal choice $q_\star$, we can straightforwardly reconstruct states $\{\rho_{S,x}^{\lambda}\}_{x\in \{0,1\}}$ from the corresponding optimizers $t_\star,z_{0,\star},z_{1,\star}$, i.e.,  
  \begin{align}
      n_{x,1} = \frac{E_x^\lambda-t_\star z_{x,\star}}{\sqrt{1-t_\star^2}}\; , \quad n_{x,2} = 0  \;,  \quad n_{x,3}= z_{x,\star} \;.
  \end{align}
If $t_\star= \pm 1$, then it is easy to verify that $z_{x,\star}= t_\star E_x$ and we obtain the desired bound by simply setting $n_{x,1}=0$.\footnote{We remark that if $t= \pm 1$, then $n_{0,3}=t E_0^\lambda$. Moreover, since $n_{0,3} \geq  1-2\omega_0^\lambda$, it must hold that $\omega_0^\lambda \geq  (1-tE_0^\lambda)/2$. In this case, Eve can simply prepare the state 
$\ket{\psi_{SE,0}^{\lambda}} = \sqrt{(1+tE_0^\lambda)/2}\,\ket{0}_{S}\otimes \ket{0}_{E} + \sqrt{(1-tE_0^\lambda)/2}\,\ket{1}_{S}\otimes \ket{1}_{E}$
which produces acceptable measurement and energy expectation values. Moreover, the measurement output has no randomness from Eve's perspective. The final expression derived in Eq.~\eqref{Eq: OptHE} also captures this, i.e., as  $t \to \pm 1$, we retrieve the required lower bound $|{\bf n}_0| \geq |E_0^\lambda |\Leftrightarrow H(S|\Lambda=\lambda,X=0) \leq  h_{2}\left((1-|E_0^\lambda|)/2\right)$.}  These vectors satisfy $|\bold{n}_0|^2 = q_{\star} \leq 1$ and $|\bold{n}_1|  \leq 1$, yield the quantum behavior $(E_0^\lambda,E_1^\lambda, \omega_0^\lambda,\omega_1^\lambda)$, and saturate the entropy bound. This ensures the bound is tight. 

  \end{proof}

\subsection{Lower Bounds on \texorpdfstring{$H(A|E\Lambda, X=0)$}{All $\lambda$}}   \label{bound-T-section}

Given a quantum behavior $(E_0,E_1,\omega_0,\omega_1)\in\mathcal{Q}$, it follows from Eq.~\eqref{objective-function-X=0} and  Theorem~\ref{theorem: optimization problem for the general model}  that:
\begin{multline} \label{Eq: OptProbCondEntAllLambda}
    {\inf}\,H(A|E\Lambda,X=0) \geq  \\
        \min_{\substack{P_\Lambda,\{(\bm{E}^{\lambda},\bm{\omega}^{\lambda})\}\\
         \{q^\lambda,t^{\lambda},z_0^{\lambda},z_1^{\lambda}\}}} 
         \sum_{\lambda}\operatorname{P}_{\Lambda}(\lambda)
        f(E_0^\lambda,q^\lambda)
\end{multline}
subject to
\begin{align}
    & && && 
        \sum_{\lambda}\operatorname{P}_{\Lambda}(\lambda)E_{x}^{\lambda}
        =E_{x}\;,\;\;\;-1\leq E_x^\lambda\leq 1 \nonumber
        \\
        & && && 
        \sum_{\lambda}\operatorname{P}_{\Lambda}(\lambda)\omega_{x}^{\lambda}
        \leq\omega_{x}\;,\;\;\;\;\;0\leq\omega_x^\lambda\leq 1 \nonumber
        \\
        & && && 
        \sum_{\lambda}\operatorname{P}_{\Lambda}(\lambda)=1
        \;,\;\;\;\;\;\;\;\;\;\;\;\;
        \operatorname{P}_{\Lambda}(\lambda)\geq 0
        \\
        & && &&
        \begin{pmatrix}
          q^\lambda & E_0^{\lambda} & z_0^\lambda \\
          E_0^{\lambda} & 1 & t^\lambda \\
          z_0^\lambda & t^\lambda & 1
        \end{pmatrix} \succeq 0 
        \; , \quad 
        \begin{pmatrix}
          1 & E_1^{\lambda} & z_1^\lambda \\
          E_1^{\lambda} & 1 & t^\lambda \\
          z_1^\lambda & t^\lambda & 1
        \end{pmatrix} \succeq 0 \nonumber
        \\
        & && && 
        1\geq z_x^\lambda \geq 1-2\omega_x^\lambda  
        \; , \quad 0\leq q^\lambda \leq 1 \nonumber 
        \\
        & && &&  
        (E^\lambda_0,E^\lambda_1,\omega^\lambda_0,\omega^\lambda_1)
        \in\mathcal{Q} \; , \nonumber
\end{align}
where $f(E_0^\lambda,q^\lambda)=h_2\left(\frac{1-E_0^\lambda}{2}\right)-h_2\left(\frac{1-\sqrt{q^\lambda}}{2}\right)$. Moreover, we note that including the upper bound constraint $\omega_x^\lambda\leq 1$ does not alter the optimization problem. This follows from the discussion Section~\ref{Sec: QubitReductions}, where we show it suffices to only consider energy operators $ \widetilde O$ whose eigenvalues are $\{0,1\}$. In particular, if $\omega_x^\lambda\geq 1$ for at least one $X=x$, then $(E^\lambda_0,E^\lambda_1,\omega^\lambda_0,\omega^\lambda_1)$ can always be generated via a convex mixture of classical, i.e., deterministic, correlations whose output will already contain no randomness from Eve's perspective~\cite[Eq.~(37)]{Van_Himbeeck_2017}.

    Section~\ref{Sec: CompSDPBounds} now discusses how this problem can be relaxed to a hierarchy of SDPs. This hierarchy yields increasingly tight certified lower bounds on the RHS of  Eq.~\eqref{Eq: OptProbCondEntAllLambda}. Afterwards, we compare our results in Section~\ref{Sec: CompClassAtt} to the classical von Neumann entropy bounds of~\cite{vanhimbeeck2019correlations}, the min-entropy bounds against quantum adversaries of~\cite{vanhimbeeck2019correlations}, and the min-entropy bounds from~\cite{ranu2026randomnesscertifiedenergyconstrainedsemideviceindependent} in the entanglement assisted setting.

\begin{rem}
        A consequence of Remark~\ref{Rem: ProjMeasReduc} and Jordan's lemma is that we simply need to consider single-qubit attacks and projective measurements; see Lemma~\ref{claim: decomposition of the operators to 3 Hilbert spaces}. However, it is not a priori clear that all elements in $\mathcal{Q}$ can be generated with qubit states and projective measurements.\footnote{Interestingly, it is possible to generate all feasible quantum behaviors on two-dimensional spaces when $M$ is described by a  POVM~\cite[Corollary 3]{vanhimbeeck2019correlations}.} As such the relaxed constraint $(E^\lambda_0,E^\lambda_1,\omega^\lambda_0,\omega^\lambda_1)\in\mathcal{Q}$ may yield lower bounds that are not tight. Nevertheless, we expect this bound to be very tight, especially when $(E_0,E_1,\omega_0,\omega_1)\in\mathcal{Q}$ is supposed to yield high entropy. To highlight why we believe this to be true, we note that we expect high von Neumann entropy to be certified when $(E_0,E_1,\omega_0,\omega_1)$ lies near an extremal point of $\mathcal{Q}$; see, e.g.,~\cite[Section~3.4]{vanhimbeeck2019correlations} where high entropy is certified near the boundary of $\mathcal{Q}$. However, all behaviors $(E_0,E_1,\omega_0,\omega_1)$ at extremal points can be generated with single-qubit attacks and projective measurements~\cite[Lemma~2]{vanhimbeeck2019correlations}. Consequently, the bounds of Theorem~\ref{theorem: optimization problem for the general model} are tight in this regime.
    \end{rem}
\subsubsection{Computable SDP Bounds} \label{Sec: CompSDPBounds}
Currently, the objective function in Eq.~\eqref{Eq: OptProbCondEntAllLambda} is a convex mixture of the function
\begin{multline} \label{Eq: ObjFctNonLinear}
    f(E_{0},q) =h_2\left(\frac{1-E_0}{2}\right)-h_2\left(\frac{1-\sqrt{q}}{2}\right) \; .
\end{multline}
The first goal of this section is to construct an appropriate piecewise linear function which (arbitrarily tightly) lower bounds $f(E_{0},q)$ such that one may use the techniques from~\cite{vanhimbeeck2019correlations,Hahn_2022} to solve the resulting optimization problem via an SDP. The fact that the linearized function is a lower bound on the original objective function then ensures that the computed value is a certified lower bound on the desired conditional entropy.

We construct said piecewise linear function using the same approach as~\cite{vanhimbeeck2019correlations,Hahn_2022}. However, we note that we rely on a modified proof as $f(E_{0},q)$ is \emph{not} jointly concave and the results of~\cite{vanhimbeeck2019correlations,Hahn_2022} therefore do not apply to our setting. 
Rather, the function $f(E_{0},q)$ is concave in $E_0$ and convex in $q$.
Our goal is to find a  partition $\mathcal{S}=\{S^{(i,j)}\}_{i,j}$ of the domain $(E,q)\in [ -1,1] \times[0,1]$ such that for each segment $S_{i}$
\begin{itemize}
  \item  there exists an affine function which lower bounds $f(E_{0},q)$ on $S^{(i,j)}$.
    \item  there exist affine constraints which are satisfied by a tuple $(E_{0},q)$ if and only if $(E_{0},q)\in S^{(i,j)}$.
\end{itemize}
Analogous to~\cite{vanhimbeeck2019correlations,Hahn_2022}, we then show that these constraints ensure that the resulting relaxed optimization problem is an SDP. To construct a partition with the desired properties, we utilize Proposition~\ref{Lem: PartitionConstruction}. 
\begin{prop} \label{Lem: PartitionConstruction}
Let $f(E,q)=h_2\left(\frac{1-E}{2}\right)-h_2\left(\frac{1-\sqrt{q}}{2}\right)$. For any integer $N\in \mathbb{N}$, let $\{E_i,q_j\}$ be a uniform grid of points on $ \left[-1,1\right] \times \left[0,1\right]$ such that
\begin{align}
    E_{i}&= \frac{2i}{N}-1 \quad &&\forall i \in \{0,1,\dots,N\} \\
    q_j&= \frac{j}{N} \quad &&\forall j \in \{0,1,\dots,N\}\; .
\end{align}
Moreover, for any $i,j \in \{0,\dots, N-1\}$, we set\footnote{For $q_0=0$, we take the limit $\lim_{q\to 0}\frac{1}{4\sqrt{q}}
    \log \left( \frac{1-\sqrt{q}}{1+\sqrt{q}}\right) = - \frac{1}{2\ln(2)}$, where $\ln$ denotes the natural logarithm.}
    \begin{multline}
        k_{i}(E) \coloneq h_2\left(\frac{1-E_{i}}{2}\right)\\
         + \frac{h_2\left(\frac{1-E_{i+1}}{2}\right)-h_2\left(\frac{1-E_{i}}{2}\right)}{E_{i+1}-E_{i}} \left(E-E_{i} \right)
    \end{multline}
\begin{multline}
    l_{j}(q) \coloneq h_2\left(\frac{1-\sqrt{q_j}}{2}\right) \\+ \frac{1}{4\sqrt{q_{j}}}
    \log \left( \frac{1-\sqrt{q_{j}}}{1+\sqrt{q_{j}}}\right)
 \left(q-q_{j}\right) \; .
\end{multline}
Then the piecewise linear function
\begin{multline} \label{Eq: PiecewiseLinFct}
    g(E,q)= k_{i}(E) - l_{j}(q) \\
    \quad \mathrm{for} \ E\in [E_i,E_{i+1}),\ q\in [q_j,q_{j+1}) 
\end{multline}
satisfies $g(E,q)\leq f(E,q)$ for all $(E,q)\in [ -1,1) \times[0,1)$.
\end{prop}

\begin{proof}
    A proof can be found in Appendix~\ref{sec:linear}
\end{proof}
One can thus lower bound the RHS of  Eq.~\eqref{Eq: OptProbCondEntAllLambda} by replacing $f(E_0^\lambda,q^\lambda)$ with the piecewise linear function $g(E_0^\lambda,q^\lambda)$ in Eq.~\eqref{Eq: PiecewiseLinFct}. 
For the remainder of this section, we consider the following partition $\mathcal{S}$ of the space $\left[-1,1\right] \times \left[0,1\right]$, which contains all tuples $(E_0,q)$. For any $N \in \mathbb{N}$, let $\mathcal{S}^{(i,j)}$ denote the partition generated via Proposition~\ref{Lem: PartitionConstruction}, i.e., $\mathcal{S}=\{ S^{(i,j)}\}_{i,j}$ where 
    \begin{align} \label{Eq: PartPlus}
        S^{(i,j)} \coloneq \{ (E,q ) \mid E\in [E_{(i)},E_{i+1}],\ q\in [q_j,q_{j+1}]\} \; .
    \end{align}
Due to continuity, the linear function from Eq.~\eqref{Eq: PiecewiseLinFct} is indeed a lower bound on $f(E_0^\lambda,q^\lambda)$ on this closed interval.
We now use the partition $\mathcal{S}$ to construct computable lower bounds on Eq.~\eqref{Eq: OptProbCondEntAllLambda} that can be solved via an SDP.  The resulting optimization is given in Theorem~\ref{Prop: SumReduction}. In particular, this bound is achieved by first using the approach from~\cite{Bancal_2014,Nieto_Silleras_2014,Hahn_2022} to further reduce the objective function of Eq.~\eqref{Eq: OptProbCondEntAllLambda} to a finite sum. Here, the number of elements in the sum is bounded by the number of elements in $\mathcal{S}$, which is denoted by $|\mathcal{S}|$.  To achieve this reduction, we use both Lemma~\ref{Lem: SDPRep}, which expresses the quantum set $\mathcal{Q}$ in terms of an SDP, and the fact that each element in $\mathcal{S}$ has a boundary which can be expressed via affine constraints. 

Afterwards, similar to~\cite{Hahn_2022} (see also~\cite{vanhimbeeck2019correlations}, which uses a dual formulation to resolve this issue),
we show there exists a change of variables for which the resulting optimization problem is an SDP. In particular, we consider the variables $\operatorname{P}_\Lambda$, $\{(\tilde{E}_0^\lambda,\tilde{E}_1^\lambda,\tilde{\omega}_0^\lambda,\tilde{\omega}_1^\lambda)\}_{
\lambda}$, and $\{\tilde{q}^\lambda,\tilde{t}^{\lambda},\tilde{z}_0^{\lambda},\tilde{z}_1^{\lambda}\}$, where $\operatorname{P}_\Lambda$ is a normalized probability distribution on $|\mathcal{S}|$ outputs. Moreover, $\{(\tilde{E}_0^\lambda,\tilde{E}_1^\lambda,\tilde{\omega}_0^\lambda,\tilde{\omega}_1^\lambda)\}_{
\lambda}$ is a family of renormalized quantum behaviors, i.e., we say $(\operatorname{P}_{\Lambda}(\lambda),\tilde{E}_0^\lambda,\tilde{E}_1^\lambda,\tilde{\omega}_0^\lambda,\tilde{\omega}_1^\lambda)\in \tilde{\mathcal{Q}}$ if and only if $ \left(\frac{\tilde{E}_0^\lambda}{\operatorname{P}_\Lambda(\lambda)},\frac{\tilde{E}_1^\lambda}{\operatorname{P}_\Lambda(\lambda)},\frac{\tilde{\omega}_0^\lambda}{\operatorname{P}_\Lambda(\lambda)},\frac{\tilde{\omega}_1^\lambda}{\operatorname{P}_\Lambda(\lambda)} \right)\in \mathcal{Q}$ is a valid quantum behavior.\footnote{If $\operatorname{P}_{\Lambda}(\lambda)=0$, then we say $(\operatorname{P}_{\Lambda}(\lambda),\tilde{E}_0^\lambda,\tilde{E}_1^\lambda,\tilde{\omega}_0^\lambda,\tilde{\omega}_1^\lambda)\in \tilde{\mathcal{Q}}$ if and only if $\tilde{E}_0^\lambda=\tilde{E}_1^\lambda=\tilde{\omega}_0^\lambda=\tilde{\omega}_1^\lambda=0$.} For any element $S^{(i,j)}\in\mathcal{S}$ of the partition, we also say $\left(\operatorname{P}_\Lambda(\lambda),\tilde{E}_0^\lambda,\tilde{q}^\lambda\right)\in \tilde{S}^{(i,j)}$ if $ \left(\frac{\tilde{E}_0^\lambda}{\operatorname{P}_\Lambda(\lambda)},\frac{\tilde{q}^\lambda}{\operatorname{P}_\Lambda(\lambda)}\right)\in S^{(i,j)}$.\footnote{Similarly, if $\operatorname{P}_{\Lambda}(\lambda)=0$, then we say $\left(\operatorname{P}_\Lambda(\lambda),\tilde{E}_0^\lambda,\tilde{q}^\lambda\right)\in \tilde{S}^{(i,j)}$ if and only if $\tilde{E}_0^\lambda=\tilde{q}^\lambda=0$.}

\begin{theorem} \label{Prop: SumReduction}
    For any $N \in \mathbb{N}$, let $\mathcal{S}=\{S_{1}, S_2,\dots\}$ denote the corresponding partition given by Eq.~\eqref{Eq: PartPlus} with $|\mathcal{S}|$ elements and let $\operatorname{P}_\Lambda$ denote a normalized probability distribution over the set $\{1,\dots,|\mathcal{S}|\} $. Then there exists a family of coefficients $\{ (\alpha^{\lambda},\beta^{\lambda},\gamma^{\lambda})\}_{\lambda\in [1,\dots,|\mathcal{S}|]}$ such that  
 \begin{widetext}
\begin{equation} \label{Eq: OptProbCondEntSDPBound}
    \begin{aligned}
        & {\inf}\,H(A|E,\Lambda,X=0) \geq && \min_{\substack{P_\Lambda,\{(\tilde{\bm{E}}^{\lambda},\tilde{\bm{\omega}}^{\lambda})\}\\
        \{\tilde{q}^\lambda,\tilde{t}^{\lambda},\tilde{z}_0^{\lambda},\tilde{z}_1^{\lambda}\}}} && \sum_{\lambda=1}^{|\mathcal{S}|} \alpha^{\lambda}\operatorname{P}_{\Lambda}(\lambda) + \beta^\lambda \tilde{E}_{0}^{\lambda}+ \gamma^\lambda \tilde{q}^{\lambda}    
        \\
        & && \text{\ \ subject to} \\
         & && && \sum_{\lambda}\tilde{E}_{x}^{\lambda} =E_{x}\;,\;\;\; \;\;\;\;  -\operatorname{P}_{\Lambda}(\lambda)\leq \tilde{E}_x^\lambda\leq \operatorname{P}_{\Lambda}(\lambda) \\
        & && && \sum_{\lambda}\tilde{\omega}_{x}^{\lambda} \leq\omega_{x}\;,\;\;\;\;\;  \;\;\;\; \;\;\;\;  0\leq\tilde{\omega}_x^\lambda\leq \operatorname{P}_{\Lambda}(\lambda)\\
        & && && \sum_{\lambda}\operatorname{P}_{\Lambda}(\lambda) =1\;,\;\;\;\;\;\;\;\;\;\;\;\;\operatorname{P}_{\Lambda}(\lambda)\geq 0\\
        & && &&
        \begin{pmatrix}
  \tilde{q}^{\lambda} & \tilde{E}_0^{\lambda} & \tilde{z}_0^\lambda \\
  \tilde{E}_0^{\lambda} & \operatorname{P}_{\Lambda}(\lambda) & \tilde{t}^\lambda \\
  \tilde{z}_0^\lambda & \tilde{t}^\lambda & \operatorname{P}_{\Lambda}(\lambda)
\end{pmatrix} \succeq 0 \; , \quad \begin{pmatrix}
  \operatorname{P}_{\Lambda}(\lambda) & \tilde{E}_1^{\lambda} & \tilde{z}_1^\lambda \\
  \tilde{E}_1^{\lambda} & \operatorname{P}_{\Lambda}(\lambda) & \tilde{t}^\lambda \\
  \tilde{z}_1^\lambda & \tilde{t}^\lambda & \operatorname{P}_{\Lambda}(\lambda)
\end{pmatrix} \succeq 0 \\
& && && \operatorname{P}_{\Lambda}(\lambda)\geq \tilde{z}_x^\lambda \geq \operatorname{P}_{\Lambda}(\lambda)-2\tilde{\omega}_x^\lambda \; , \quad 0\leq \tilde{q}^\lambda \leq \operatorname{P}_{\Lambda}(\lambda) \\
        & && &&  (\operatorname{P}_{\Lambda}(\lambda),\tilde{E}_0^\lambda,\tilde{E}_1^\lambda,\tilde{\omega}_0^\lambda,\tilde{\omega}_1^\lambda)\in \tilde{\mathcal{Q}} \;,\;\;\;\;\;\left(\operatorname{P}_\Lambda(\lambda),\tilde{E}_0^\lambda,\tilde{q}^\lambda\right)\in \tilde{S}_{\lambda}\; .
    \end{aligned}
    \end{equation}
\end{widetext}
\end{theorem}
\begin{proof}
    A proof can be found in Appendix~\ref{sec:sdp}
\end{proof}
As can be seen in Appendix~\ref{sec:sdp}, the coefficients $\{ \alpha^{\lambda},\beta^{\lambda},\gamma^{\lambda}\}_{\lambda}$ directly correspond to the coefficients of the piecewise linear function given in Eq.~\eqref{Eq: PiecewiseLinFct}, i.e.,
\begin{align}
    g(E,q)= \alpha^\lambda + \beta^\lambda E +\gamma^\lambda q\quad \mathrm{for} \ (E,
    q)\in S_{\lambda} \; .
\end{align}
Moreover, the resulting optimization problem from Theorem~\ref{Prop: SumReduction} can be computed via an SDP, using the following SDP expressions for the sets $\tilde{\mathcal{Q}} $ and $\tilde{S}_{\lambda}$.

\begin{figure*}[t]
    \centering

    \begin{minipage}[t]{0.43\textwidth}
        \centering
        \includegraphics[width=\linewidth]
        {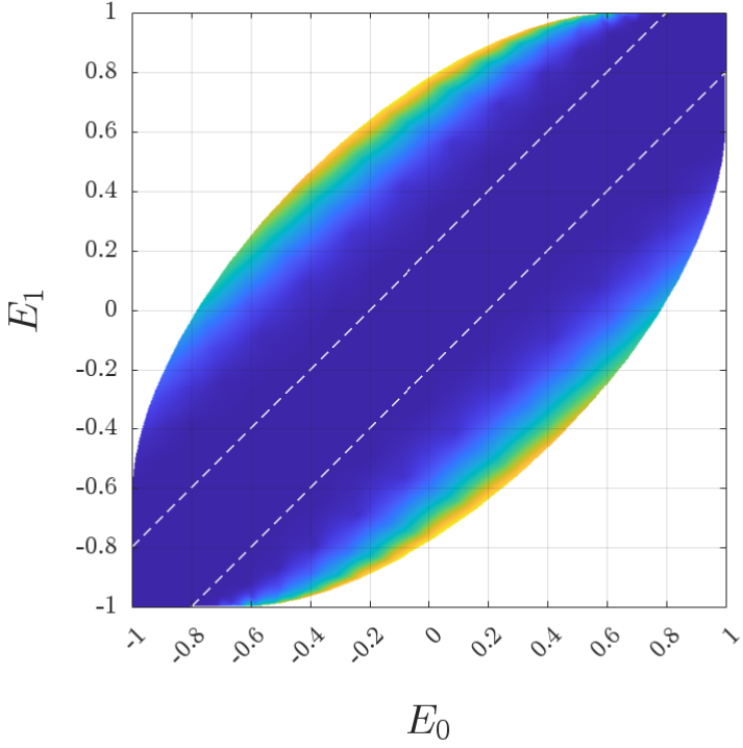}

        \vspace{1mm}
    \makebox[\linewidth][c]{%
        \hspace*{0.8cm}\textbf{(a) Classical Adversary}%
    }
    \end{minipage}
    \hfill
    \begin{minipage}[t]{0.526\textwidth}
        \centering
        \includegraphics[width=\linewidth]
        {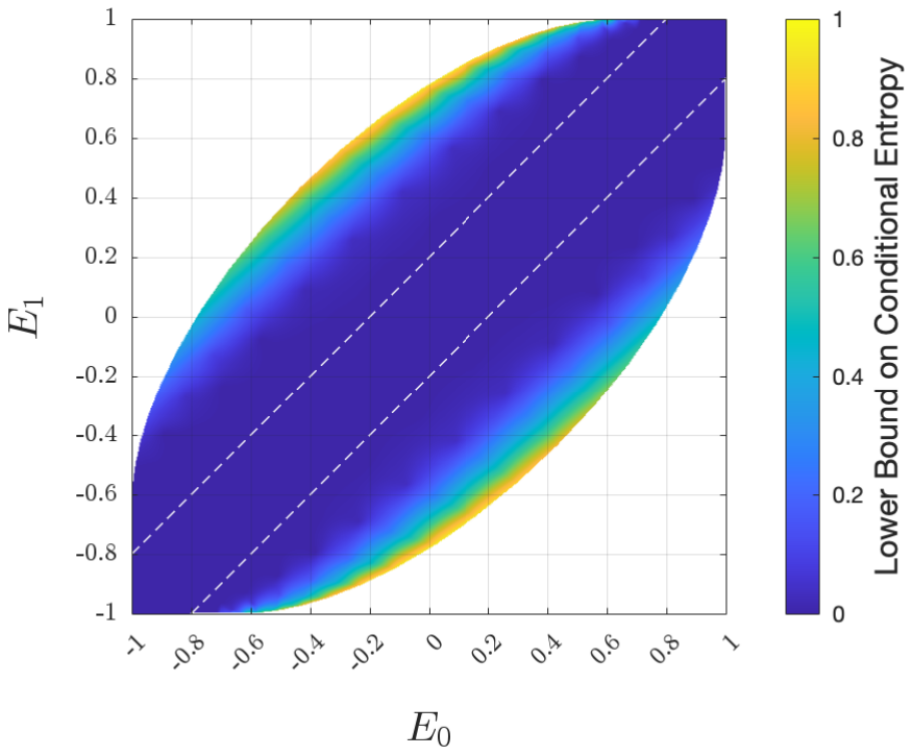}

        \makebox[\linewidth][c]{%
        \hspace*{-0.8cm}\textbf{(b) Quantum Adversary}%
    }
    \end{minipage}

    \caption{
    Lower bounds on the conditional von Neumann entropy for
    $\omega_0=\omega_1=0.05$ and $X=0$.
    In~(a) we assume Eve holds solely a classical register, i.e.,
    we lower bound $H(A|\Lambda,X=0)$.
    General lower bounds on $H(A|E\Lambda,X=0)$ against a quantum
    adversary are presented in~(b).
    The dashed lines denote the set of correlations that can be
    generated classically, i.e., deterministically~\cite{Van_Himbeeck_2017}.
    }
    \label{fig:0.05-side-by-side}
\end{figure*}
\begin{lemma}
    The tuple $(\operatorname{P}_{\Lambda}(\lambda),\tilde{E}_0^\lambda,\tilde{E}_1^\lambda,\tilde{\omega}_0^\lambda,\tilde{\omega}_1^\lambda)$ belongs to the set $\tilde{\mathcal{Q}} $ if and only if there exist $\tilde{u},\tilde{v},\tilde{\eta}_0,\tilde{\eta}_1\in\mathbb{R}$ satisfying $\tilde{\eta}_{0,1}\leq \tilde{\omega}_{0,1}$ such that the matrix
     \begin{equation}
    \begin{aligned}
       \begin{pmatrix}\operatorname{P}_{\Lambda}(\lambda)&\tilde{u}&\tilde{E}_0&2\tilde{\eta}_0-\operatorname{P}_{\Lambda}(\lambda)\\
        \tilde{u}&\operatorname{P}_{\Lambda}(\lambda)& \tilde{E}_1&2\tilde{\eta}_1-\operatorname{P}_{\Lambda}(\lambda)\\\tilde{E}_0& \tilde{E}_1&\operatorname{P}_{\Lambda}(\lambda)&\tilde{v}\\2\tilde{\eta}_0-\operatorname{P}_{\Lambda}(\lambda)&2\tilde{\eta}_1-\operatorname{P}_{\Lambda}(\lambda)&\tilde{v}&\operatorname{P}_{\Lambda}(\lambda)
        \end{pmatrix}
    \end{aligned}
    \end{equation}
    is positive semi-definite.
\end{lemma}

\begin{proof}
For $\operatorname{P}_{\Lambda}(\lambda)=0$, note that $(\operatorname{P}_{\Lambda}(\lambda),\tilde{E}_0^\lambda,\tilde{E}_1^\lambda,\tilde{\omega}_0^\lambda,\tilde{\omega}_1^\lambda)$ is in $\mathcal{Q}$ if and only if $\tilde{E}_0^\lambda=\tilde{E}_1^\lambda=\tilde{\omega}_0^\lambda=\tilde{\omega}_1^\lambda=0$. Setting $\tilde{u}=\tilde{v}=\tilde{\eta}_{0}=\tilde{\eta}_{1}=0$ trivially satisfies the positive semi-definite constraint.
 Otherwise, recall that $(\operatorname{P}_{\Lambda}(\lambda),\tilde{E}_0^\lambda,\tilde{E}_1^\lambda,\tilde{\omega}_0^\lambda,\tilde{\omega}_1^\lambda) \in \tilde{\mathcal{Q}} $ if and only if $ \left(\frac{\tilde{E}_0^\lambda}{\operatorname{P}_\Lambda(\lambda)},\frac{\tilde{E}_1^\lambda}{\operatorname{P}_\Lambda(\lambda)},\frac{\tilde{\omega}_0^\lambda}{\operatorname{P}_\Lambda(\lambda)},\frac{\tilde{\omega}_1^\lambda}{\operatorname{P}_\Lambda(\lambda)} \right)\in \mathcal{Q}$. It then follows from Lemma~\ref{Lem: SDPRep} that $(\operatorname{P}_{\Lambda}(\lambda),\tilde{E}_0^\lambda,\tilde{E}_1^\lambda,\tilde{\omega}_0^\lambda,\tilde{\omega}_1^\lambda) \in \tilde{\mathcal{Q}} $ if and only if
    there exist $u,v,\eta_0,\eta_1\in\mathbb{R}$ satisfying $\eta_{0,1}\leq \frac{\tilde{\omega}_{0,1}}{\operatorname{P}_\Lambda(\lambda)}$ such that
    \begin{equation} \label{Eq: IntermStepRenormQSet}
    \begin{aligned}
       \begin{pmatrix}1&u&\frac{\tilde{E}_0^\lambda}{\operatorname{P}_\Lambda(\lambda)}&2\eta_0-1\\
        u&1& \frac{\tilde{E}_1^\lambda}{\operatorname{P}_\Lambda(\lambda)}&2\eta_1-1\\\frac{\tilde{E}_0^\lambda}{\operatorname{P}_\Lambda(\lambda)}& \frac{\tilde{E}_1^\lambda}{\operatorname{P}_\Lambda(\lambda)}&1&v\\2\eta_0-1&2\eta_1-1&v&1
        \end{pmatrix}
    \end{aligned}
    \end{equation}
    is positive semi-definite. Multiplying Eq.~\eqref{Eq: IntermStepRenormQSet} by $\operatorname{P}_{\Lambda}(\lambda)$ and setting $\tilde{u}=\operatorname{P}_{\Lambda}(\lambda)u$, $\tilde{v}=\operatorname{P}_{\Lambda}(\lambda)v$, $\tilde{\eta}_{0,1}=\operatorname{P}_{\Lambda}(\lambda)\eta_{0,1}$ yields the desired claim.
\end{proof}
The boundary of $\tilde{S}_{\lambda}$ is characterized via affine constraints. In particular, it follows from Eq.~\eqref{Eq: PartPlus}  that $\tilde{S}_{\lambda}$ has the following form: there exist constants $E_i$, $E_{i+1}$, $q_j$, $q_{j+1}$ such that 
    \begin{multline}
        \tilde{S}_{\lambda}   \coloneq \Bigl\{\left(\operatorname{P}_\Lambda(\lambda),\tilde{E}^\lambda,\tilde{q}^\lambda\right) \mid \operatorname{P}_\Lambda(\lambda)E_i \leq 
    \tilde{E}^\lambda \leq \\
    \operatorname{P}_\Lambda(\lambda)E_{i+1},\ \operatorname{P}_\Lambda(\lambda)q_j \leq \tilde{q}^\lambda \leq \operatorname{P}_\Lambda(\lambda)q_{j+1} 
      \Bigl\} \; .
    \end{multline}
    
\subsubsection{Comparison Between Classical, Quantum, and Entanglement-Assisted Entropy Bounds} \label{Sec: CompClassAtt}
 We are now ready to compare the SDP bounds obtained via  Theorem~\ref{Prop: SumReduction} with the results of~\cite[Section 3.3]{vanhimbeeck2019correlations} and~\cite[Section~4]{ranu2026randomnesscertifiedenergyconstrainedsemideviceindependent}.
        The lower bounds on the von Neumann entropy against a quantum adversary are generated via Theorem~\ref{Prop: SumReduction} using either  $N=16$ (yielding a partition with $16^2=256$ elements) or $N=12$ (yielding a partition with $12^2=144$ elements). For  quantum behaviors near the boundary, we set $N=16$ to more tightly capture the high-entropy regions. Otherwise,  we set $N=12$. 
        Similarly, to apply the results of~\cite[Section 3.3]{vanhimbeeck2019correlations} and lower bound the von Neumann entropy for classical adversaries, we approximate the function $h_2\left(\frac{1-E_0}{2}\right)$ with a piecewise linear function $g(E_0)$ that is constructed to be linear on the $N=16$ intervals $E_0 \in [-1+2i/N,-1+2(i+1)/N]$, where $i \in \{0,\dots,15\} $.
        The results are presented  in Fig.~\ref{fig:0.05-side-by-side}.
\begin{figure*}[t]
    \centering

    \begin{minipage}[t]{0.43\textwidth}
        \vspace{0pt}
        \centering
        \includegraphics[height=6.6cm,keepaspectratio]
        {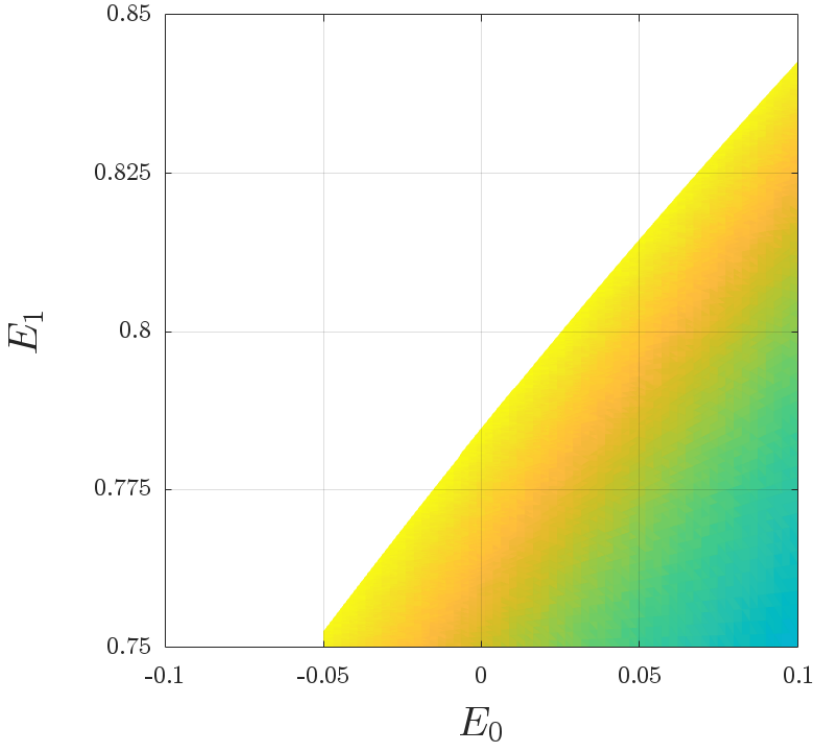}

        \vspace{1mm}
        \makebox[\linewidth][c]{%
            \hspace*{0.6cm}\textbf{(a) Classical Adversary}%
        }
    \end{minipage}
    \hfill
    \begin{minipage}[t]{0.525\textwidth}
        \vspace{0pt}
        \centering
        \includegraphics[height=6.6cm,keepaspectratio]
        {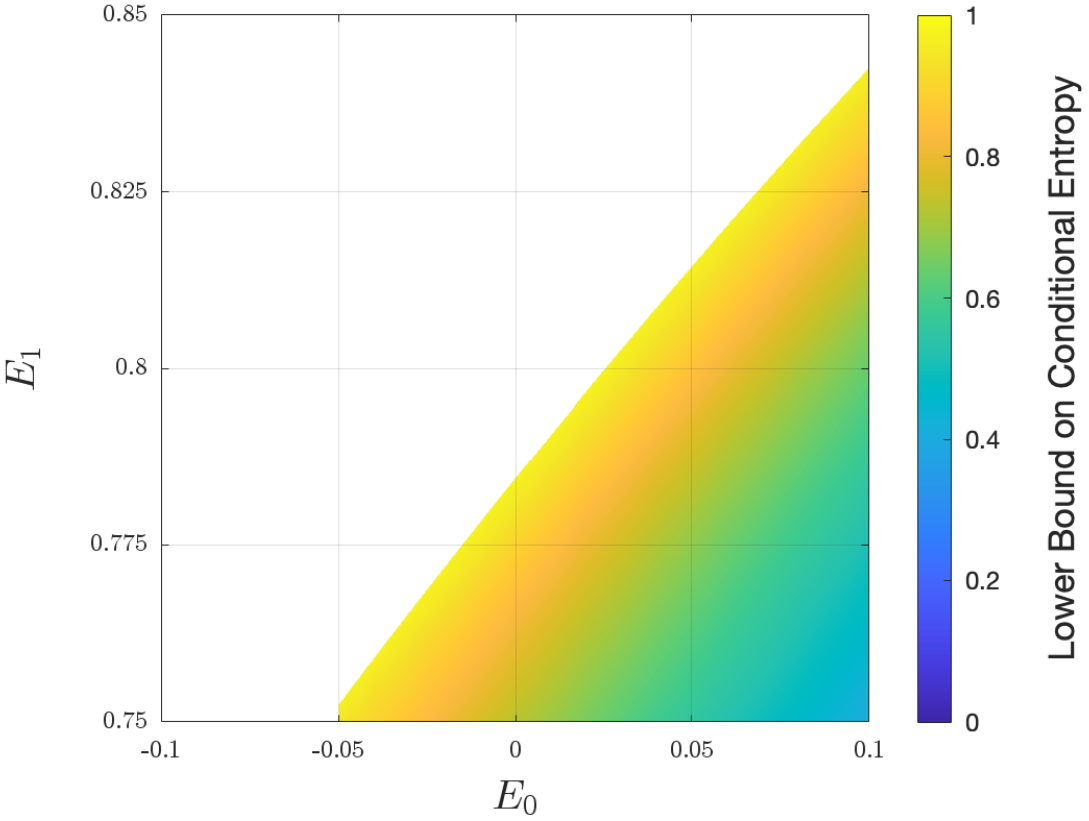}

        \vspace{1mm}
        \makebox[\linewidth][c]{%
            \hspace*{-0.6cm}\textbf{(b) Quantum Adversary}%
        }
    \end{minipage}

    \caption{
    Contains results of Fig.~\ref{fig:0.05-side-by-side}
    for a restricted, zoomed-in parameter range.
    This range contains high certifiable entropy for
    $\omega_0=\omega_1=0.05$ and $X=0$.
    }
    \label{fig:0.05-side-by-sideZoom}
\end{figure*}

For $\omega_0=\omega_1=0.05$, Fig.~\ref{fig:0.05-side-by-side} gives lower bounds on $H(A|\Lambda, X=0)$ and $H(A|E\Lambda, X=0)$ for all feasible choices of $(E_0,E_1)$. In both cases, $H(A|E\Lambda, X=0)$ appears maximal near $E_0 \simeq 0$ (see also Fig.~\ref{fig:0.05-side-by-sideZoom}). This is expected, as the key-generating measurement is $X=0$. $E_0 =0$ then corresponds to the key-generation output being uniformly random (if one does not have access to the registers $E\Lambda$). Moreover, we see that the maximal certifiable randomness for a classical adversary is $1$ bit. For this energy regime, we find that the  additional quantum side-information is quite limited by the low energy constraint. This is why we find that, for $\omega_0=\omega_1=0.05$, the 
 largest gap between the classical and quantum bounds is approximately $\lesssim 0.027$.

\begin{figure}[htbp]
  \centering
  \includegraphics[width=0.49\textwidth]{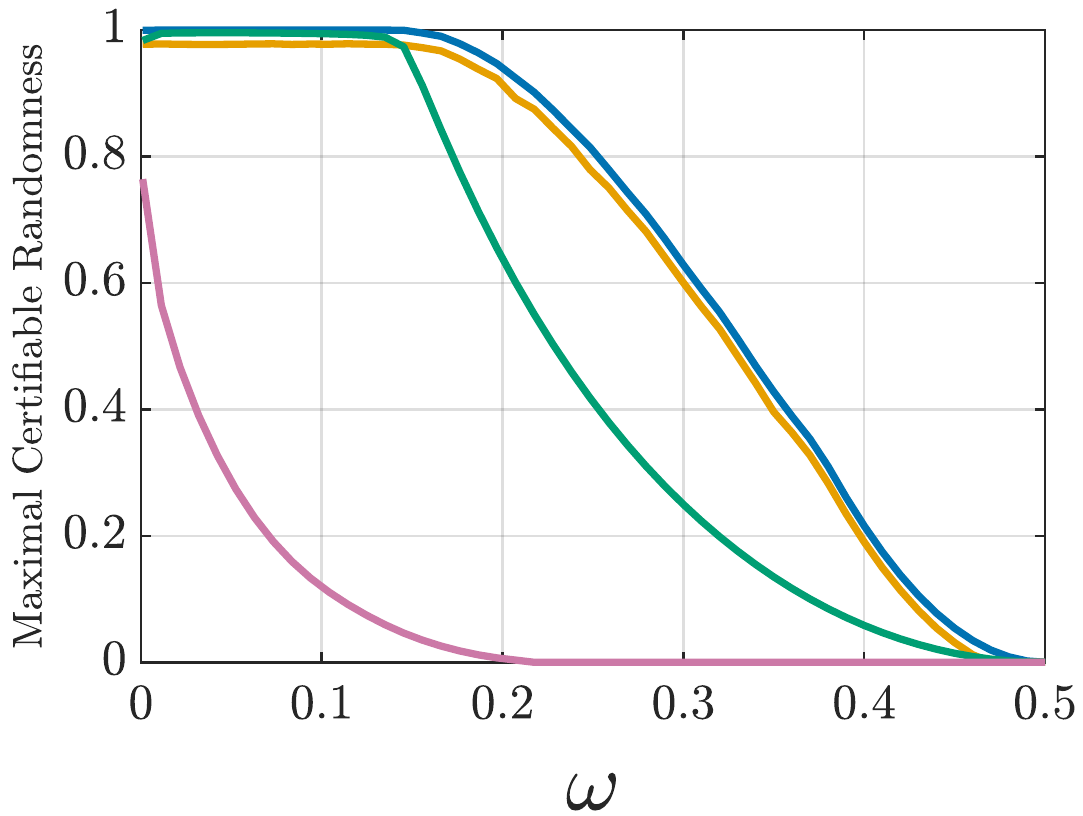} 
  
\caption{The maximal certifiable randomness as a function of $\omega=\omega_0=\omega_1$. The blue and orange lines denote the certifiable lower bounds on the von Neumann entropy  against  classical~\cite{vanhimbeeck2019correlations} and quantum adversaries (Theorem~\ref{Prop: SumReduction}), respectively. The green and purple lines denote lower boundds on the min-entropy against quantum~\cite{vanhimbeeck2019correlations} and entanglement-assisted adversaries~\cite{ranu2026randomnesscertifiedenergyconstrainedsemideviceindependent}. 
}
  \label{max-entropy-omega}
\end{figure}

For all $0< \omega=\omega_0=\omega_1 \leq 0.5$,
Fig.~\ref{max-entropy-omega} shows the maximal randomness that can be certified  when maximizing over all $(E_0,E_1,\omega,\omega) \in \mathcal{Q}$ for fixed $\omega$.
The blue and orange lines denote the maximal certified von Neumann entropy against classical~\cite{vanhimbeeck2019correlations} and quantum adversaries, respectively. The entropy bound shown in Fig.~\ref{max-entropy-omega} is computed at $N=16$ for both of these optimizations. The green line denotes the maximally certified lower bound on the min-entropy from~\cite{vanhimbeeck2019correlations}, which holds against quantum adversaries; see~\cite[Remark 6]{vanhimbeeck2019correlations}. The corresponding SDP runs optimally at $N=2$, i.e., the partition only requires two elements. Lastly, the purple line denotes the corresponding min-entropy bound against entanglement-assisted adversaries from~\cite{ranu2026randomnesscertifiedenergyconstrainedsemideviceindependent}. The resulting SDP  achieved via the NPA relaxation~\cite{Pironio_2010,Navascu_s_2008} is computed using the code from~\cite{ranu_hmin_energy_const_sdi} at NPA Level $2$.

Since the entropy bound for any given quantum behavior is computed via an SDP  and $\mathcal{Q}$ is convex, there exists an optimal choice of $(E_0,E_1,\omega,\omega)$ that lies on the boundary of $\mathcal{Q}$. For this $(E_0,E_1,\omega,\omega) \in \mathcal{Q}$, we then find that the lower bound on the von Neumann entropy against quantum adversaries closely matches the bound against classical eavesdroppers. Since the bounds of~\cite{vanhimbeeck2019correlations} are generally expected to be very tight, the near-coincidence of the two curves provides strong numerical evidence that Theorem~\ref{Prop: SumReduction} yields a very tight lower bound in this regime.  
Moreover, the closeness of the bounds    strongly suggests that the optimal attack for this quantum behavior requires Eve to send pure states. In this case, there can namely be no useful quantum side-information.

We also highlight that it is in principle possible to generate larger bounds on the maximally certifiable min-entropy in the entanglement-assisted setting for a fixed $\omega$. This is due to the fact, in this setting, the set of achievable correlations is strictly larger than $\mathcal{Q}$~\cite{DAvino:2025ymd}. While it is therefore possible to optimize over a larger set and achieve potentially tighter bounds,
we choose to always optimize over $\mathcal{Q}$ to achieve an honest comparison between the four bounds. 

We conclude  by noting that there is, in principle, a stark discontinuity at $\omega=0$. At $\omega=0$, Eve can namely only send the ground state, independently of the input choice $X=x$ and classical side-information $\Lambda=\lambda$. Crucially, this enforces that $E_0^\lambda = E_1^\lambda$ for every $\Lambda=\lambda$. It is then easy to see from Eq.~\eqref{Q-criterion-arcsin} that Eve can generate any behavior $(E,E,0,0)$ with $-1\leq E \leq 1$ using solely deterministic strategies. For $\omega> 0$, this is no longer the case.

\section{Discussion and Open Questions} \label{Section: DiscussionOpenQ}
In this work, we study the simplest semi device-independent prepare and measure set-up from~\cite{Van_Himbeeck_2017,vanhimbeeck2019correlations}, which uses energy constraints to ensure that the underlying measurement device acts as a QRNG. Concretely, we generalize the numerical entropy bounds introduced in~\cite{vanhimbeeck2019correlations} to now account for adversaries with access to quantum side-information. 
To achieve this desired SDP formulation in Theorem~\ref{Prop: SumReduction}, we first apply Jordan's lemma~\cite{pironio2009deviceindependent} and derive lower bounds on the conditional entropy $H(A|E,\Lambda=\lambda, X=0)_{\rho}$ for qubit systems; see Theorem~\ref{theorem: optimization problem for the general model}. We then apply the techniques of~\cite{Bancal_2014,Nieto_Silleras_2014,Hahn_2022} to derive \textit{computable} lower bounds on the resulting expression for $H(A|E\Lambda, X=0)_{\rho}$. The resulting bounds on $H(A|E\Lambda, X=0)_{\rho}$  for general attacks are extremely tight and outperform the min-entropy bounds for large ranges of $\omega$; see Fig.~\ref{max-entropy-omega}.

In Section~\ref{Sec: CompClassAtt}, we study two regimes in which one may expect the conditional von Neumann entropy against quantum adversaries to closely match the classical bound. For small energy regimes $\omega_x \ll 1$, Eve has to essentially be sending states that are close to the ground state. Since the ground state is pure, there can be no meaningful quantum side-information. As can be seen in Fig.~\ref{fig:0.05-side-by-side} and Fig.~\ref{fig:0.05-side-by-sideZoom}, our computed bounds respect this intuition. Similarly, the entropy for the optimal quantum behavior with a fixed energy parameter $\omega$ is plotted in Fig.~\ref{max-entropy-omega}. The behavior has to lie on the boundary of $\mathcal{Q}$ and our results imply that the optimal attack requires pure states. Our techniques can account for this phenomenon and, unlike the min-entropy bounds of~\cite{vanhimbeeck2019correlations,ranu2026randomnesscertifiedenergyconstrainedsemideviceindependent}, our SDP hierarchy seemingly converges to the classical von Neumann bounds of~\cite{vanhimbeeck2019correlations}.

The approach of~\cite{ranu2026randomnesscertifiedenergyconstrainedsemideviceindependent} certifies randomness against entanglement-assisted adversaries and is very versatile, in the sense that it can handle general set-ups and can, in principle, be adapted to the conditional von Neumann entropy, and conditional Petz-Rényi entropies using the techniques of~\cite{Brown_2024, hahn2024,kossmann2025reliableentropyestimationobserved}. However, as already noted in~\cite{ranu2026randomnesscertifiedenergyconstrainedsemideviceindependent}, it appears to give extremely loose bounds on the min-entropy. We believe this is primarily due to the fact that, for quantum adversaries, the non-degeneracy constraint of the energy operator is not a priori directly compatible with an NPA relaxation~\cite{Pironio_2010,Navascu_s_2008}.\footnote{For classical adversaries this is indeed directly implementable using the NPA hierarchy; see~\cite{Roch_i_Carceller_2025}} As such,~\cite{ranu2026randomnesscertifiedenergyconstrainedsemideviceindependent} have to replace the non-degeneracy constraint with weaker constraints. Since the approach taken in this work does not require a weakening of the energy constraint assumption, deriving analogous bounds on the von Neumann in the entanglement-assisted setting remains an interesting open question. We believe a similar approach to the one take in this work may be able to provide analogous tight bounds in the entanglement-assisted setting. In particular, such a comparison would give insights as to the tightness of the results derived in~\cite{ranu2026randomnesscertifiedenergyconstrainedsemideviceindependent,li2026onchipquantumrandomnessamplification}.

\section*{Code Availability}
The Python code used to generate the data can be found at:
\url{https://github.com/Thomas0501/EnergyConstraintsSDPs}. The python code in this work relies on  MOSEK~\cite{mosek}, CVXPY~\cite{diamond2016cvxpy,agrawal2018rewriting}  and NCPOL2SDPA~\cite{Wittek_2015} to solve the resulting SDPs.

\bibliography{References.bib}
\newpage

\onecolumngrid
\section*{Appendix}

\section{SDP Formulation for $\Lambda=\lambda$ Entropy Bound} \label{App: OptimalValueQ}
Recall that we are considering the following two optimization problems. First, we have the optimization problem
\begin{equation}\label{App: MinNormN}
    \begin{aligned}
         & \min_{t,\{y_x,z_x\}} && y_0^2  + z_0^2 \\
        & \text{subject to} \\
        & && 1\geq z_x \geq 1-2\omega_x \ \forall x\in \{0,1\} \; , \quad -1 \leq t \leq 1 \;, \\
         & && E_x = \sqrt{1-t^2}y_x +tz_x \ \forall x\in \{0,1\}\;  ,  \quad y_x^2  + z_x^2  \leq 1 \ \forall x\in \{0,1\}\; .
    \end{aligned}
    \end{equation}
Secondly, we desire to express this via the SDP
\begin{equation} \label{App: OptTildeQ}
    \begin{aligned}
        &  && \min_{q,t,z_0,z_1} &&q\\
        & && \text{subject to} \\
         & && && \begin{pmatrix}
  q & E_0 & z_0 \\
  E_0 & 1 & t \\
  z_0 & t & 1
\end{pmatrix} \succeq 0 \; , \quad \begin{pmatrix}
  1 & E_1 & z_1 \\
  E_1 & 1 & t \\
  z_1 & t & 1
\end{pmatrix} \succeq 0\; , \\
        & && && 1\geq z_x \geq 1-2\omega_x \ \forall x\in \{0,1\} \; , \quad 0\leq q \leq 1 \;.
    \end{aligned}
    \end{equation}
\begin{theorem}
    For any quantum behavior $(E_0,E_1,\omega_0,\omega_1)\in\mathcal{Q}$, the optimization problem from 
    Eq.~\eqref{App: MinNormN} 
    is equivalent to Eq.~\eqref{App: OptTildeQ} 
\end{theorem}
\begin{proof}
    First, we prove that for any feasible point $(t,\{y_{x},z_{x}\})$ of Eq.~\eqref{App: MinNormN},   $(q,t,z_{0},z_{1})$ is a feasible point of Eq.~\eqref{App: OptTildeQ}, where $q = y_0^2 + z_0^2 $. Then, we will show the converse, i.e., 
for any feasible point $(q,t,z_{0},z_{1})$ of Eq.~\eqref{App: OptTildeQ} there exists a feasible point $(t,\{y_{x},z_{x}\})$ of Eq.~\eqref{App: MinNormN} such that $q \geq  y_0^2  + z_0^2 $.
    Together, these then imply that the optimization problems 
    Eq.~\eqref{App: MinNormN} 
    and Eq.~\eqref{App: OptTildeQ} equivalent.

    Let $(t,\{y_{x},z_{x}\})$ satisfy 
    \begin{multline}
        1\geq z_x \geq 1-2\omega_x \ \forall x\in \{0,1\} \; , \quad -1 \leq t \leq 1 \;, \\  E_x = \sqrt{1-t^2}y_x +tz_x \ \forall x\in \{0,1\}\;  ,  \quad y_x^2  + z_x^2  \leq 1 \ \forall x\in \{0,1\}\; .
    \end{multline}
    Then, setting  $q = y_0^2  + z_0^2 $, $(q,t,z_{0},z_{1})$ trivially satisfy the constraints
    \begin{align}
        1\geq z_x \geq 1-2\omega_x \ \forall x\in \{0,1\} \; , \quad 0\leq q \leq 1 \; .
    \end{align}
We therefore simply need to verify the positive semi-definite constraints are satisfied.  Following Sylvester's criterion (see, e.g.,~\cite[Chapter~7.6]{10.1137/1.9781611977448}), the positive semi-definite constraints  are satisfied if all principal minors are non-negative. That is, 
\begin{align}
     \begin{pmatrix}
  q & E_0 & z_0 \\
  E_0 & 1 & t \\
  z_0 & t & 1
\end{pmatrix} \succeq 0  \quad \begin{pmatrix}
  1 & E_1 & z_1 \\
  E_1 & 1 & t \\
  z_1 & t & 1
\end{pmatrix} \succeq 0
\end{align}
are, respectively, equivalent to 
\begin{align} \label{App: PrincipalMinors}
    q \geq 0 \; , \quad &q \geq E_0^2\; , \quad q \geq z_0^2 \; , \quad  |t| \leq 1 \; , \quad q\left(1-t^2\right) -E_0^2 -z_0^2 + 2E_0z_0 t \geq 0 \\
     &1 \geq E_1^2\; , \quad 1 \geq z_1^2 \; , \quad  |t| \leq 1 \; , \quad 1\left(1-t^2\right) -E_1^2 -z_1^2 + 2E_1z_1 t \geq 0
    \; .
\end{align}
Both the fourth inequality of the first line and the first three inequalities of the second line hold trivially. Additionally, the choice $q = y_0^2  + z_0^2 $ naturally satisfies the first and third constraint of the first line. The second inequality of the first line follows from the Cauchy-Schwarz inequality, i.e., 
\begin{align}
    E_0 = \left(y_0,z_0\right) \cdot \begin{pmatrix}
        \sqrt{1-t^2} \\
        t
    \end{pmatrix} \implies E_0^2 \leq \left( y_0^2 +z_0^2\right) \cdot \left( (1-t^2)+t^2\right) = q \; .
\end{align}
 The first line's last inequality can be proven to hold exactly. This can be proven as follows. For $q = y_0^2  + z_0^2 $, we first have that 
\begin{align}
    q\left(1-t^2\right) -E_0^2 -z_0^2 + 2E_0z_0 t &= \left(y_0^2+z_0^2\right)\left(1-t^2\right) -E_0^2 -z_0^2 + 2E_0z_0 t \\
     &= \left(y_0^2+z_0^2\right)\left(1-t^2\right) - (E_0-tz_0)^2-z_0^2(1-t^2) \\
    & = -(E_0-tz_0)^2+y_0^2(1-t^2) \; .
\end{align}
Using the energy constraint condition, $E_0 = \sqrt{1-t^2}y_0 +tz_0$, we then find that 
\begin{align}
    -(E_0-tz_0)^2+y_0^2(1-t^2) &= -(1-t^2)y_0^2+ (1-t^2)y_0^2 \\
    &=0 \; .
\end{align}
The last inequality of the second line can be proven analogously.

\begin{align}
 \left(1-t^2\right) -E_1^2 -z_1^2 + 2E_1z_1 t 
     &= \left(1-t^2\right) - (E_1-tz_1)^2-z_1^2(1-t^2) \\
    & = -(E_1-tz_1)^2+(1-z_1^2)(1-t^2) \\
    & = -(1-t^2)y_1^2+(1-z_1^2)(1-t^2) \\
    & = (1-y_1^2-z_1^2)(1-t^2) \\
    &\geq 0
    \; .
\end{align}
Here, the last line follows from the fact that $y_1^2+z_1^2\leq 1$.

We now show the reverse direction. 
Let $(q,t,z_{0},z_{1})$ satisfy 
    \begin{equation} 
    \begin{aligned}
         & && && \begin{pmatrix}
  q & E_0 & z_0 \\
  E_0 & 1 & t \\
  z_0 & t & 1
\end{pmatrix} \succeq 0 \; , \quad \begin{pmatrix}
  1 & E_1 & z_1 \\
  E_1 & 1 & t \\
  z_1 & t & 1
\end{pmatrix} \succeq 0\; , \\
        & && && 1\geq z_x \geq 1-2\omega_x \ \forall x\in \{0,1\} \; , \quad 0\leq q \leq 1 \;.
    \end{aligned}
    \end{equation}
Then, by Eq.~\eqref{App: PrincipalMinors}, we have that 
\begin{align}
     q\left(1-t^2\right) -E_0^2 -z_0^2 + 2E_0z_0 t \geq 0 \;, \quad 
      \left(1-t^2\right) -E_1^2 -z_1^2 + 2E_1z_1 t \geq 0
    \; .
\end{align}
These conditions can alternatively be expressed as
\begin{align}
     q\left(1-t^2\right) - (E_0-tz_0)^2-z_0^2(1-t^2) \geq 0 \;, \quad 
      \left(1-t^2\right) - (E_1-tz_1)^2-z_1^2(1-t^2) \geq 0
    \; .
\end{align}
Let us now consider the first inequality. From Eq.~\eqref{App: PrincipalMinors}, we have that $q \geq z_0^2$; moreover at $q = z_0^2$ the LHS can no longer be positive. Then, due to monotonicity, there must exist a value $y_0^2$  such that 
\begin{align}
    1\geq q \geq y_0^2+z_0^2 \quad \mathrm{and} \quad \left(y_0^2+z_0^2\right)\left(1-t^2\right) - (E_0-tz_0)^2-z_0^2(1-t^2) = 0 \; .
\end{align}
Similarly, it follows from the second inequality that there must exist  a $y_1^2$ such that
\begin{align}
    1 \geq y_1^2+z_1^2 \quad \mathrm{and} \quad \left(y_1^2+z_1^2\right)\left(1-t^2\right) - (E_1-tz_1)^2-z_1^2(1-t^2) = 0 \; .
\end{align}
Using the same analysis as above, we show these expressions are equivalent to energy constraints, i.e.,
\begin{align}
&\left(y_x^2+z_x^2\right)\left(1-t^2\right) - (E_x-tz_x)^2-z_x^2(1-t^2) = 0 \\\Longleftrightarrow &  y_x^2\left(1-t^2\right) - (E_x-tz_x)^2 = 0  \\
    \Longleftrightarrow& |y_x|\sqrt{1-t^2} = |E_x-tz_x| \; .
\end{align}
Up to now, we had a degree of freedom with regards to the sign of $y_x$; we now set our choice of $y_x$ such that
\begin{align}
    E_x=\sqrt{1-t^2}y_x +tz_x 
\end{align}
holds. The resulting tuple $(t,\{y_{x},z_{x}\})$ is a feasible point  of Eq.~\eqref{App: MinNormN}  and satisfies $q \geq y_0^2+z_0^2$, as desired.
\end{proof}

\section{Conditional Entropy: SDP Approximation}
In this appendix we approximate the optimization problem we have obtained with an SDP. 
Our starting point is Eq.~\eqref{Eq: OptProbCondEntAllLambda} for the minimal value for the conditional von-Neumann entropy, conditioned on the input $X=0$, which we repeat here for convenience:
\begin{equation} \label{general-bound-T-appendix}
    \begin{aligned}
        & {\inf}\,H(A|E\Lambda,X=0) \geq && \inf_{\substack{P_\Lambda,\{(\bm{E}^{\lambda},\bm{\omega}^{\lambda})\}\\
         \{q^\lambda,t^{\lambda},z_0^{\lambda},z_1^{\lambda}\}}} && \sum_{\lambda}\operatorname{P}_{\Lambda}(\lambda)f(E_0^\lambda,q^\lambda)\\
        & && \text{\ \ subject to} \\
         & && && \sum_{\lambda}\operatorname{P}_{\Lambda}(\lambda)E_{x}^{\lambda} =E_{x}\;,\;\;\;-1\leq E_x^\lambda\leq 1 \\
        & && && \sum_{\lambda}\operatorname{P}_{\Lambda}(\lambda)\,\omega_{x}^{\lambda} \leq\omega_{x}\;,\;\;\;\;\;0\leq\omega_x^\lambda\leq 1\\
        & && && \sum_{\lambda}\operatorname{P}_{\Lambda}(\lambda) =1\;,\;\;\;\;\;\;\;\;\;\;\;\;\operatorname{P}_{\Lambda}(\lambda)\geq 0\\
        & && &&
        \begin{pmatrix}
  q^\lambda & E_0^{\lambda} & z_0^\lambda \\
  E_0^{\lambda} & 1 & t^\lambda \\
  z_0^\lambda & t^\lambda & 1
\end{pmatrix} \succeq 0 \; , \quad \begin{pmatrix}
  1 & E_1^{\lambda} & z_1^\lambda \\
  E_1^{\lambda} & 1 & t^\lambda \\
  z_1^\lambda & t^\lambda & 1
\end{pmatrix} \succeq 0 \\
& && && 1\geq z_x^\lambda \geq 1-2\omega_x^\lambda \ \forall x\in \{0,1\} \; , \quad 0\leq q^\lambda \leq 1 \\
 & && &&  (E^\lambda_0,E^\lambda_1,\omega^\lambda_0,\omega^\lambda_1)\in\mathcal{Q} \; ,\\
    \end{aligned}
    \end{equation}
   where $f(E_0^\lambda,q^\lambda)=h_2\left(\frac{1-E_0^\lambda}{2}\right)-h_2\left(\frac{1-\sqrt{q^\lambda}}{2}\right)$.
    In Appendix \ref{sec:linear}, we  construct a piecewise linear function that lower bounds the objective function. This is later used in Appendix \ref{sec:sdp} to replace our optimization problem with an SDP.

\subsection{Piecewise Linear Lower Bound Construction}\label{sec:linear}

\begin{prop} \label{App: PartitionConstruction}
Let $f(E,q)=h_2\left(\frac{1-E}{2}\right)-h_2\left(\frac{1-\sqrt{q}}{2}\right)$. For any integer $N\in \mathbb{N}$, let $\{E_i,q_j\}$ be a uniform grid of points on $ \left[-1,1\right] \times \left[0,1\right]$ such that
\begin{align}
    E_{i}&= \frac{2i}{N}-1 \quad &&\forall i \in \{0,1,\dots,N\} \\
    q_j&= \frac{j}{N} \quad &&\forall j \in \{0,1,\dots,N\}\; .
\end{align}
Moreover, for any $i,j\in \{0,\dots, N-1\}$, we set\footnote{For $q_0=0$, we take the limit $\lim_{q\to 0}\frac{1}{4\sqrt{q}}
    \log \left( \frac{1-\sqrt{q}}{1+\sqrt{q}}\right) = - \frac{1}{2\ln(2)}$, where $\ln$ denotes the natural logarithm.}
\begin{align}
    k_{i}(E) &\coloneq h_2\left(\frac{1-E_{i}}{2}\right) + \frac{h_2\left(\frac{1-E_{i+1}}{2}\right)-h_2\left(\frac{1-E_{i}}{2}\right)}{E_{i+1}-E_{i}} \left(E-E_{i} \right)\\
    l_{j}(q) &\coloneq h_2\left(\frac{1-\sqrt{q_j}}{2}\right)+ \frac{1}{4\sqrt{q_{j}}}
    \log \left( \frac{1-\sqrt{q_{j}}}{1+\sqrt{q_{j}}}\right)
 \left(q-q_{j}\right) \; .
\end{align}
Then the piecewise linear function
\begin{align} \label{App: PiecewiseLinFct}
    g(E,q)= k_{i}(E) - l_{j}(q) \quad \mathrm{for} \ E\in [E_i,E_{i+1}),\ q\in [q_j,q_{j+1}) 
\end{align}
satisfies $g(E,q)\leq f(E,q)$ for all $(E,q)\in [ -1,1) \times[0,1)$.
\end{prop}

\begin{proof}
For any $i \in \{0,1,\dots,N-1\}$ and $E\in [E_i,E_{i+1}] $ 
there exists a $\lambda \in [0,1]$ such that 
\begin{align} \label{Eq: ConvEExpress}
    E = \left(1-\lambda\right) E_i + \lambda E_{i+1} \; .
\end{align}
    It then follows from the fact that $h_2\left(\frac{1-E}{2}\right)$ is concave in  $E$ that     
    \begin{align} 
        h_2\left(\frac{1-E}{2}\right) &\geq \left(1-\lambda\right)h_2\left(\frac{1-E_i }{2}\right) +\lambda h_2\left(\frac{1-E_{i+1} }{2}\right) \\
        &= h_2\left(\frac{1-E_i }{2}\right) +\lambda \left( h_2\left(\frac{1-E_{i+1}}{2}\right)-h_2\left(\frac{1-E_{i}}{2}\right)\right) \\
        &= h_2\left(\frac{1-E_{i}}{2}\right) + \frac{h_2\left(\frac{1-E_{i+1}}{2}\right)-h_2\left(\frac{1-E_{i}}{2}\right)}{E_{i+1}-E_{i}} \left(E-E_{i} \right) \\
        &\eqcolon  k_{i}(E) \; .
    \end{align}
The second-to-last equality holds due to Eq.~\eqref{Eq: ConvEExpress}, which can alternatively be expressed as 
\begin{align}
    \lambda = \frac{E-E_{i}}{E_{i+1}-E_{i}}
\end{align}

Let us now consider the second term $ h_2\left(\frac{1-\sqrt{q}}{2}\right)$. This function is concave in $q$. This can be seen as a consequence of $h_2\left(\frac{1+\sqrt{1-q}}{2}\right)$ being concave~\cite{Coffman_2000}. Alternatively, it is readily verifiable that
the second derivative of $ h_2\left(\frac{1-\sqrt{q}}{2}\right)$ is given by
\begin{align}
   \frac{d^2}{dq^2} h_2\left(\frac{1-\sqrt{q}}{2}\right) = - \frac{\frac{\sqrt{q}}{1-q}-\mathrm{artanh}(\sqrt{q})}{4q^{3/2}\ln 2}  \leq 0\; .
\end{align}
The inequality simply follows from the fact that
\begin{align}
    \mathrm{artanh}(\sqrt{x}) = \int_{0}^{\sqrt{x}} \frac{1}{1-u^2}du \leq \int_{0}^{\sqrt{x}} \frac{1}{1-x}du = \frac{\sqrt{x}}{1-x} \; .
\end{align}

For any  $j \in \{0,1,\dots,N-1\}$, the construction of $l_{j}(q)$ satisfies 
\begin{align} \label{Eq: ConstructionLProp}
h_2\left(\frac{1-\sqrt{q_j}}{2}\right)=l_{j}(q_{j}) \quad \mathrm{and} \quad \frac{d}{dq}h_2\left(\frac{1-\sqrt{q_j}}{2}\right)|_{q=q_{j}}=l_{j}^\prime(q_{j})\; .
\end{align}
Moreover, it follows from the concavity of $h_2\left(\frac{1-\sqrt{q}}{2}\right)$ and linearity of $l_{j}(q)$ that 
\begin{align} \label{Eq: SecondDerProp}
   \left(\frac{d}{dq}h_2\left(\frac{1-\sqrt{q}}{2}\right)-l_{j}^{ \prime}(q)\right)^\prime = \frac{d^2}{dq^2}h_2\left(\frac{1-\sqrt{q}}{2}\right)-l_{j}^{\prime  \prime}(q) = \frac{d^2}{dq^2}h_2\left(\frac{1-\sqrt{q}}{2}\right) \leq 0 \quad \mathrm{for}\ q \in [q_{j},q_{j+1}) \; .
\end{align}
Together, Eq.~\eqref{Eq: ConstructionLProp} and Eq.~\eqref{Eq: SecondDerProp} imply that 
\begin{align}\label{Eq: FirstDerProp}
\frac{d}{dq}h_2\left(\frac{1-\sqrt{q}}{2}\right)-l_{j}^{ \prime}(q) \leq 0 \quad \mathrm{for}\ q \in [q_{j},q_{j+1})\; .
\end{align}
    Similarly, Eq.~\eqref{Eq: ConstructionLProp} and Eq.~\eqref{Eq: FirstDerProp} then  ensure that
    \begin{align}
        h_2\left(\frac{1-\sqrt{q}}{2}\right)-l_{j}(q) \leq 0 \quad \mathrm{for}\ q \in [q_{j},q_{j+1}) \; .
    \end{align}
    This proves that $l_{j}(q)$ is an upper bound on $ h_2\left(\frac{1-\sqrt{q}}{2}\right)$. For every $E\in [E_i,E_{i+1}),\ q\in [q_j,q_{j+1})$, it therefore follows that
    \begin{align}
        f(E,q) \geq k_{i}(E)-l_{j}(q)\; ,
    \end{align}
    thus concluding the proof.
\end{proof} 

\subsection{SDP Hierarchy for the von Neumann Entropy}\label{sec:sdp}
For completeness, we recall from Section~\ref{Sec: CompSDPBounds} the following notation. We say $(\operatorname{P}_{\Lambda}(\lambda),\tilde{E}_0^\lambda,\tilde{E}_1^\lambda,\tilde{\omega}_0^\lambda,\tilde{\omega}_1^\lambda)\in \tilde{\mathcal{Q}}$ if and only if $ \left(\frac{\tilde{E}_0^\lambda}{\operatorname{P}_\Lambda(\lambda)},\frac{\tilde{E}_1^\lambda}{\operatorname{P}_\Lambda(\lambda)},\frac{\tilde{\omega}_0^\lambda}{\operatorname{P}_\Lambda(\lambda)},\frac{\tilde{\omega}_1^\lambda}{\operatorname{P}_\Lambda(\lambda)} \right)\in \mathcal{Q}$ is a valid quantum behavior. For any element $S\in\mathcal{S}$ of the partition, we also say $\left(\operatorname{P}_\Lambda(\lambda),\tilde{E}_0^\lambda,\tilde{q}^\lambda\right)\in \tilde{S}$ if $ \left(\frac{\tilde{E}_0^\lambda}{\operatorname{P}_\Lambda(\lambda)},\frac{\tilde{q}^\lambda}{\operatorname{P}_\Lambda(\lambda)}\right)\in S$.

\begin{theorem} \label{App: SumReduction}
    For any $n \in \mathbb{N}$, let $\mathcal{S}=\{S_{1}, S_2,\dots\}$ denote the corresponding partition given by Eq.~\eqref{Eq: PartPlus} with $|\mathcal{S}|$ elements and let $\operatorname{P}_\Lambda$ denote a normalized probability distribution over the set $\{1,\dots,|\mathcal{S}|\} $. Then there exists a family of coefficients $\{ (\alpha^{\lambda},\beta^{\lambda},\gamma^{\lambda})\}_{\lambda\in [1,\dots,|\mathcal{S}|]}$ such that
\begin{equation} \label{App: OptProbCondEntSDPBound}
    \begin{aligned}
        & {\inf}\,H(A|E,\Lambda,X=0) \geq && \min_{\substack{P_\Lambda,\{(\tilde{\bm{E}}^{\lambda},\tilde{\bm{\omega}}^{\lambda})\}\\
        \{\tilde{q}^\lambda,\tilde{t}^{\lambda},\tilde{z}_0^{\lambda},\tilde{z}_1^{\lambda}\}}} && \sum_{\lambda=1}^{|\mathcal{S}|} \alpha^{\lambda}\operatorname{P}_{\Lambda}(\lambda) + \beta^\lambda \tilde{E}_{0}^{\lambda}+ \gamma^\lambda \tilde{q}^{\lambda}    
        \\
        & && \text{\ \ subject to} \\
         & && && \sum_{\lambda}\tilde{E}_{x}^{\lambda} =E_{x}\;,\;\;\; \;\;\;\;  -\operatorname{P}_{\Lambda}(\lambda)\leq \tilde{E}_x^\lambda\leq \operatorname{P}_{\Lambda}(\lambda) \\
        & && && \sum_{\lambda}\tilde{\omega}_{x}^{\lambda} \leq\omega_{x}\;,\;\;\;\;\;  \;\;\;\; \;\;\;\;  0\leq\tilde{\omega}_x^\lambda\leq \operatorname{P}_{\Lambda}(\lambda)\\
        & && && \sum_{\lambda}\operatorname{P}_{\Lambda}(\lambda) =1\;,\;\;\;\;\;\;\;\;\;\;\;\;\operatorname{P}_{\Lambda}(\lambda)\geq 0\\
        & && &&
        \begin{pmatrix}
  \tilde{q}^{\lambda} & \tilde{E}_0^{\lambda} & \tilde{z}_0^\lambda \\
  \tilde{E}_0^{\lambda} & \operatorname{P}_{\Lambda}(\lambda) & \tilde{t}^\lambda \\
  \tilde{z}_0^\lambda & \tilde{t}^\lambda & \operatorname{P}_{\Lambda}(\lambda)
\end{pmatrix} \succeq 0 \; , \quad \begin{pmatrix}
  \operatorname{P}_{\Lambda}(\lambda) & \tilde{E}_1^{\lambda} & \tilde{z}_1^\lambda \\
  \tilde{E}_1^{\lambda} & \operatorname{P}_{\Lambda}(\lambda) & \tilde{t}^\lambda \\
  \tilde{z}_1^\lambda & \tilde{t}^\lambda & \operatorname{P}_{\Lambda}(\lambda)
\end{pmatrix} \succeq 0 \\
& && && \operatorname{P}_{\Lambda}(\lambda)\geq \tilde{z}_x^\lambda \geq \operatorname{P}_{\Lambda}(\lambda)-2\tilde{\omega}_x^\lambda  \; , \quad 0\leq \tilde{q}^\lambda \leq \operatorname{P}_{\Lambda}(\lambda) \\
        & && &&  (\operatorname{P}_{\Lambda}(\lambda),\tilde{E}_0^\lambda,\tilde{E}_1^\lambda,\tilde{\omega}_0^\lambda,\tilde{\omega}_1^\lambda)\in \tilde{\mathcal{Q}} \;,\;\;\;\;\;\left(\operatorname{P}_\Lambda(\lambda),\tilde{E}_0^\lambda,\tilde{q}^\lambda\right)\in \tilde{S}_{\lambda}\; .
    \end{aligned}
    \end{equation}  
\end{theorem}
\begin{proof}
    From Eq.~\eqref{Eq: OptProbCondEntAllLambda}, we have the following lower bound on the desired von Neumann entropy:
\begin{equation} \label{App: OptProbCondEntAllLambda}
    \begin{aligned}
        & {\inf}\,H(A|E\Lambda,X=0) \geq && \min_{\substack{P_\Lambda,\{(\bm{E}^{\lambda},\bm{\omega}^{\lambda})\}\\
         \{q^\lambda,t^{\lambda},z_0^{\lambda},z_1^{\lambda}\}}} && \sum_{\lambda}\operatorname{P}_{\Lambda}(\lambda)f(E_0^\lambda,q^\lambda)\\
        & && \text{\ \ subject to} \\
         & && && \sum_{\lambda}\operatorname{P}_{\Lambda}(\lambda)E_{x}^{\lambda} =E_{x}\;,\;\;\;-1\leq E_x^\lambda\leq 1 \\
        & && && \sum_{\lambda}\operatorname{P}_{\Lambda}(\lambda)\,\omega_{x}^{\lambda} \leq\omega_{x}\;,\;\;\;\;\;0\leq\omega_x^\lambda\leq 1\\
        & && && \sum_{\lambda}\operatorname{P}_{\Lambda}(\lambda) =1\;,\;\;\;\;\;\;\;\;\;\;\;\;\operatorname{P}_{\Lambda}(\lambda)\geq 0\\
        & && &&
        \begin{pmatrix}
  q^\lambda & E_0^{\lambda} & z_0^\lambda \\
  E_0^{\lambda} & 1 & t^\lambda \\
  z_0^\lambda & t^\lambda & 1
\end{pmatrix} \succeq 0 \; , \quad \begin{pmatrix}
  1 & E_1^{\lambda} & z_1^\lambda \\
  E_1^{\lambda} & 1 & t^\lambda \\
  z_1^\lambda & t^\lambda & 1
\end{pmatrix} \succeq 0 \\
& && && 1\geq z_x^\lambda \geq 1-2\omega_x^\lambda \; , \quad 0\leq q^\lambda \leq 1 \\
 & && &&  (E^\lambda_0,E^\lambda_1,\omega^\lambda_0,\omega^\lambda_1)\in\mathcal{Q} \; ,\\
    \end{aligned}
    \end{equation} 
where $f(E_0^\lambda,q^\lambda)=h_2\left(\frac{1-E_0^\lambda}{2}\right)-h_2\left(\frac{1-\sqrt{q^\lambda}}{2}\right)$. Moreover, for any $S_{k}\in \mathcal{S}$, it follows from Proposition~\ref{Lem: PartitionConstruction} and the continuity of $f$ that there exist constants $\alpha^k, \beta^k,\gamma^k$ such that  
\begin{align} \label{App: CoeffFam}
    f(E,q) \geq  \alpha^k + \beta^k E +\gamma^k q\quad \mathrm{for} \ (E,
    q)\in S_{k} \; .
\end{align}
Let us now consider any feasible point in the optimization from Eq.~\eqref{App: OptProbCondEntAllLambda}:
\begin{align} \label{App: GenFeasPoint}
  \{P_\Lambda(\lambda),(\bm{E}^{\lambda},\bm{\omega}^{\lambda}),{q}^\lambda,{t}^{\lambda},{z}_0^{\lambda},{z}_1^{\lambda}\}_\lambda  \; .
\end{align} 
Following~\cite[Proof of Proposition 9]{Hahn_2022}, let  
\begin{align} \label{App: RKConstruction}
   \mathcal{R}_{k} = \{ \lambda| (E_0^\lambda,q^\lambda) \in S_k \ \land \ \lambda \notin \mathcal{R}_{l} \ \mathrm{for} \ l<k\} \; .
\end{align}
The condition $\lambda \notin \mathcal{R}_{l} \ \mathrm{for} \ l<k$ ensures that we include 
all $\Lambda =\lambda$ for which $(E_0^\lambda,q^\lambda)$ is an element in  $S_k$, except for those that are also contained in another partition element $S_l$ with $l<k$. We highlight that, similar to~\cite[Proof of Proposition 9]{Hahn_2022}, if $(E_0^\lambda,q^\lambda)$ lies on an edge between two or more partitions, it does not matter for the following proof in which $\mathcal{R}_{k}$ we include $\lambda$. It is only important that each $\lambda$ appears \textit{exactly} once, i.e., no element is over- or under-counted. 
We then have that
\begin{align}   \label{App: GenAttackLoweBound}
\sum_{\lambda}\operatorname{P}_{\Lambda}(\lambda)f(E_0^\lambda,q^\lambda) \geq \sum_{k} \sum_{\lambda \in \mathcal{R}_{k}} \operatorname{P}_{\Lambda}(\lambda)\left(\alpha^k + \beta^k E^\lambda +\gamma^k q^\lambda \right) \; .
\end{align}
We now use the approach of~\cite{Bancal_2014,Nieto_Silleras_2014,Hahn_2022} to reduce the number of $\{P_\Lambda(\lambda),(\bm{E}^{\lambda},\bm{\omega}^{\lambda}),{q}^\lambda,{t}^{\lambda},{z}_0^{\lambda},{z}_1^{\lambda}\}$
to at most one per partition element $S_{k}$. Concretely, for any feasible point of the form Eq~\eqref{App: GenFeasPoint}, we will now construct a different point
 \begin{align} \label{App: ReducedFeasiblePoint}
     \{P^\prime(k),(\bm{E}^{\prime k},\bm{\omega}^{\prime k}),{q}^{\prime k},{t}^{\prime k},{z}_0^{\prime k},{z}_1^{\prime k}\}_{k \in \{ 1,\dots,|\mathcal{S}|\}}
 \end{align}
with $|\mathcal{S}|$ elements that is also feasible and yields the same value for the RHS of Eq.~\eqref{App: GenAttackLoweBound}. 

We set
\begin{align}
    P^\prime(k) =\sum_{\lambda \in \mathcal{R}_{k}} P_\Lambda(\lambda) \; , \quad (\bm{E}^{\prime k},\bm{\omega}^{\prime k}) = \frac{1}{P^\prime(k)}\sum_{\lambda \in \mathcal{R}_{k}}(P_\Lambda(\lambda)\bm{E}^{ \lambda},P_\Lambda(\lambda)\bm{\omega}^{ \lambda})\; , \\
    ({q}^{\prime k},{t}^{\prime k},{z}_0^{\prime k},{z}_1^{\prime k}) = \frac{1}{P^\prime(k)}\sum_{\lambda \in \mathcal{R}_{k}} (P_\Lambda(\lambda)q^{ \lambda},P_\Lambda(\lambda)t^{ \lambda},P_\Lambda(\lambda){z}_0^{ \lambda},P_\Lambda(\lambda){z}_1^{ \lambda}) \; .
\end{align}
If $P^\prime(k)=0$, then one may set $(\bm{E}^{\prime k},\bm{\omega}^{\prime k}),{q}^{\prime k},{t}^{\prime k},{z}_0^{\prime k},{z}_1^{\prime k}$
to any (feasible) value. Let us therefore, without loss of generality, restrict ourselves to $P^\prime(k)>0$. By construction, the point from Eq.~\eqref{App: ReducedFeasiblePoint} satisfies the constraints
\begin{align}
    \sum_{k}\operatorname{P}(k)E_{x}^{\prime k} &=\sum_{\lambda}\operatorname{P}_{\Lambda}(\lambda)E_{x}^{\lambda}=E_{x} \; ,\\ \sum_{k}\operatorname{P}(k)\,\omega_{x}^{\prime k} &= \sum_{\lambda}\operatorname{P}_{\Lambda}(\lambda)\,\omega_{x}^{\lambda} \leq\omega_{x} \; , \\
    \sum_{k}\operatorname{P}(k) &= \sum_{\lambda}\operatorname{P}_{\Lambda}(\lambda) = 1 \; .
\end{align}
Moreover, 
we have that
\begin{align}
    -1\leq E_x^{\prime k}\leq 1\; , \quad 0\leq \omega_x^{\prime k}\leq 1 \; , \quad 1\geq z_x^{\prime k} \geq 1-2\omega_x^{\prime k} \; , 0\leq q^{\prime k}\leq 1 \; .
\end{align}
This can be seen from the fact that 
\begin{align}
    -1 = \frac{ \sum_{\lambda \in \mathcal{R}_{k}} P_\Lambda(\lambda)  \cdot \left( -1\right)}{P^\prime(k)} \leq \frac{ \sum_{\lambda \in \mathcal{R}_{k}} P_\Lambda(\lambda)  \cdot \left(E_x^{\lambda}\right)}{P^\prime(k)} = E_x^{\prime k}  \leq \frac{ \sum_{\lambda \in \mathcal{R}_{k}} P_\Lambda(\lambda)  \cdot \left(+1 \right)}{P^\prime(k)} \leq 1  \; ,\\ 
    0 = \frac{ \sum_{\lambda \in \mathcal{R}_{k}} P_\Lambda(\lambda)  \cdot \left( 0\right)}{P^\prime(k)} \leq \frac{ \sum_{\lambda \in \mathcal{R}_{k}} P_\Lambda(\lambda)  \cdot \left(\omega_x^{\lambda}\right)}{P^\prime(k)} = \omega_x^{\prime k}  \leq \frac{ \sum_{\lambda \in \mathcal{R}_{k}} P_\Lambda(\lambda)  \cdot \left(+1 \right)}{P^\prime(k)} \leq 1 \; ,\\
    1 = \frac{ \sum_{\lambda \in \mathcal{R}_{k}} P_\Lambda(\lambda)  \cdot \left( +1 \right)}{P^\prime(k)}  \geq  \frac{ \sum_{\lambda \in \mathcal{R}_{k}} P_\Lambda(\lambda)  \cdot \left( z_x^\lambda \right)}{P^\prime(k)}  =  z_x^{\prime k} \geq   \frac{ \sum_{\lambda \in \mathcal{R}_{k}} P_\Lambda(\lambda)  \cdot \left( 1-2\omega_x^{\lambda} \right)}{P^\prime(k)} = 1-2\omega_x^{\prime k} \; , \\
    0 = \frac{ \sum_{\lambda \in \mathcal{R}_{k}} P_\Lambda(\lambda)  \cdot \left( 0\right)}{P^\prime(k)} \leq \frac{ \sum_{\lambda \in \mathcal{R}_{k}} P_\Lambda(\lambda)  \cdot \left( q^\lambda\right)}{P^\prime(k)} = q^{\prime k} \leq \frac{ \sum_{\lambda \in \mathcal{R}_{k}} P_\Lambda(\lambda)  \cdot \left( 1\right)}{P^\prime(k)} = 1 \; .
\end{align}
Lastly, we have that $(E^{\prime k}_0,E^{\prime k}_1,\omega^{\prime k}_0,\omega^{\prime k}_1)\in\mathcal{Q}$ and 
\begin{align}
     \begin{pmatrix}
  q^{\prime k} & E_0^{\prime k} & z_0^{\prime k} \\
  E_0^{\prime k} & 1 & t^{\prime k} \\
  z_0^{\prime k} & t^{\prime k} & 1
\end{pmatrix} \succeq 0 \; , \quad \begin{pmatrix}
  1 & E_1^{\prime k} & z_1^{\prime k} \\
  E_1^{\prime k} & 1 & t^{\prime k} \\
  z_1^{\prime k} & t^{\prime k} & 1
\end{pmatrix} \succeq 0 \; .
\end{align}
The first follows from the fact that $\mathcal{Q}$
 is closed under convex mixtures~\cite[Section~4.1]{Van_Himbeeck_2017}  and the fact that  
\begin{align}
    \begin{pmatrix}
  q^{\prime k} & E_0^{\prime k} & z_0^{\prime k} \\
  E_0^{\prime k} & 1 & t^{\prime k} \\
  z_0^{\prime k} & t^{\prime k} & 1
\end{pmatrix}  = \frac{1}{P^\prime(k)} \left(\sum_{\lambda \in \mathcal{R}_{k}} P_\Lambda(\lambda) \begin{pmatrix}
  q^{\lambda} & E_0^{\lambda} & z_0^{\lambda} \\
  E_0^{\lambda} & 1 & t^{\lambda} \\
  z_0^{\lambda} & t^{\lambda} & 1
\end{pmatrix}\right)\succeq 0 \; ,\\
 \begin{pmatrix}
  1 & E_1^{\prime k} & z_1^{\prime k} \\
  E_1^{\prime k} & 1 & t^{\prime k} \\
  z_1^{\prime k} & t^{\prime k} & 1
\end{pmatrix}  = \frac{1}{P^\prime(k)} \left(\sum_{\lambda \in \mathcal{R}_{k}} P_\Lambda(\lambda) \begin{pmatrix}
  1 & E_1^{\lambda} & z_1^{\lambda} \\
  E_1^{\lambda} & 1 & t^{\lambda} \\
  z_1^{\lambda} & t^{\lambda} & 1
\end{pmatrix}\right)\succeq 0 \; .
\end{align}
The last operator inequality follows from the fact that $P_\Lambda(\lambda) \geq 0$ and
\begin{align}
     \begin{pmatrix}
  q^{\lambda} & E_0^{\lambda} & z_0^{\lambda} \\
  E_0^{\lambda} & 1 & t^{\lambda} \\
  z_0^{\lambda} & t^{\lambda} & 1
\end{pmatrix} \succeq 0 \; , \quad \begin{pmatrix}
  1 & E_1^{\lambda} & z_1^{\lambda} \\
  E_1^{\lambda} & 1 & t^{\lambda} \\
  z_1^{\lambda} & t^{\lambda} & 1
\end{pmatrix} \succeq 0 \quad \forall\lambda \in \mathcal{R}_{k}\; .
\end{align}
Together, this proves that Eq.~\eqref{App: ReducedFeasiblePoint} is a feasible point. To show that this feasible point achieves the same value as the RHS of Eq.~\eqref{App: GenAttackLoweBound}, we note that
 \begin{align}
     &\sum_{k} P^\prime(k)\left(\alpha^k + \beta^k E^{\prime k} +\gamma^k q^{\prime k}\right) \\=& \sum_{k} P^\prime(k)\left(\alpha^k + \beta^k \frac{\left(\sum_{\lambda \in \mathcal{R}_{k}} P_\Lambda(\lambda)E_0^{\lambda }\right)}{P^\prime(k)} +\gamma^k \frac{\left(\sum_{\lambda \in \mathcal{R}_{k}} P_\Lambda(\lambda)q^{\lambda }\right)}{P^\prime(k)}\right)\\
=& \sum_{k} \sum_{\lambda \in \mathcal{R}_{k}} \operatorname{P}_{\Lambda}(\lambda)\left(\alpha^k + \beta^k E^\lambda +\gamma^k q^\lambda \right)\; .
 \end{align}
This proves it suffices to consider normalized probability distributions  $\operatorname{P}_\Lambda$ over the set $\{1,\dots,|\mathcal{S}|\} $ and the existence of a family of coefficients $\{ (\alpha^{\lambda},\beta^{\lambda},\gamma^{\lambda})\}_{\lambda\in [1,\dots,|\mathcal{S}|]}$ such that
\begin{equation} \label{App: SemiFinalOptProb}
    \begin{aligned}
        & {\inf}\,H(A|E,\Lambda,X=0) \geq && \min_{\substack{P_\Lambda,\{({\bm{E}}^{\lambda},{\bm{\omega}}^{\lambda})\}\\
        \{{q}^\lambda,{t}^{\lambda},z_0^{\lambda},{z}_1^{\lambda}\}}} && \sum_{\lambda=1}^{|\mathcal{S}|} \alpha^{\lambda}\operatorname{P}_{\Lambda}(\lambda) + \beta^\lambda \operatorname{P}_{\Lambda}(\lambda){E}_{0}^{\lambda}+ \gamma^\lambda \operatorname{P}_{\Lambda}(\lambda){q}^{\lambda}    
        \\
        & && \text{\ \ subject to} \\
         & && && \sum_{\lambda}\operatorname{P}_{\Lambda}(\lambda){E}_{x}^{\lambda} =E_{x}\;,\;\;\; \;\;\;\;  -1\leq {E}_x^\lambda\leq 1 \\
        & && && \sum_{\lambda}\operatorname{P}_{\Lambda}(\lambda){\omega}_{x}^{\lambda} \leq\omega_{x}\;,\;\;\;\;\;  \;\;\;\; \;\;\;\;  0\leq{\omega}_x^\lambda\leq 1\\
        & && && \sum_{\lambda}\operatorname{P}_{\Lambda}(\lambda) =1\;,\;\;\;\;\;\;\;\;\;\;\;\;\operatorname{P}_{\Lambda}(\lambda)\geq 0\\
        & && &&
        \begin{pmatrix}
  {q}^{\lambda} & {E}_0^{\lambda} & {z}_0^\lambda \\
  {E}_0^{\lambda} & 1 & {t}^\lambda \\
  {z}_0^\lambda & {t}^\lambda & 1
\end{pmatrix} \succeq 0 \; , \quad \begin{pmatrix}
  1 & {E}_1^{\lambda} & {z}_1^\lambda \\
  {E}_1^{\lambda} & 1 & {t}^\lambda \\
  {z}_1^\lambda & {t}^\lambda & 1
\end{pmatrix} \succeq 0 \\
& && && 1\geq {z}_x^\lambda \geq 1-2{\omega}_x^\lambda  \; , \quad 0\leq {q}^\lambda \leq 1\\
        & && &&  ({E}_0^\lambda,{E}_1^\lambda,{\omega}_0^\lambda,{\omega}_1^\lambda)\in {\mathcal{Q}} \;,\;\;\;\;\;\left({E}_0^\lambda,{q}^\lambda\right)\in {S}_{\lambda}\; .
    \end{aligned}
    \end{equation}  
    To arrive at Eq.~\eqref{App: OptProbCondEntSDPBound} and conclude the proof, we now simply apply the following transformation of variables 
   \begin{align}
       \{P_\Lambda(\lambda),(\bm{E}^{\lambda},\bm{\omega}^{\lambda}),{q}^\lambda,{t}^{\lambda},{z}_0^{\lambda},{z}_1^{\lambda}\}_\lambda \leftrightarrow \{P_\Lambda(\lambda),(\tilde{\bm{E}}^{\lambda},\tilde{\bm{\omega}}^{\lambda}),\tilde{q}^\lambda,\tilde{t}^{\lambda},\tilde{z}_0^{\lambda},\tilde{z}_1^{\lambda}\}_\lambda \; ,
   \end{align}
   where 
   \begin{multline}
(\tilde{\bm{E}}^{\lambda},\tilde{\bm{\omega}}^{\lambda}) \coloneq (P_\Lambda(\lambda){\bm{E}}^{\lambda},P_\Lambda(\lambda){\bm{\omega}}^{\lambda})  \; , \quad  \tilde{q}^\lambda \coloneq P_\Lambda(\lambda){q}^\lambda \; , \quad \tilde{t}^\lambda \coloneq P_\Lambda(\lambda){t}^\lambda \\
     \tilde{z}_0^\lambda \coloneq P_\Lambda(\lambda){z}_0^\lambda \; , \quad \tilde{z}_1^\lambda \coloneq P_\Lambda(\lambda){z}_1^\lambda \; .
   \end{multline}
    This results in an equivalent  optimization problem, as this mapping is reversible unless $P_\Lambda(\lambda)=0$. 
    However, whenever  $P_\Lambda(\lambda)=0$,
    \begin{align}
       \left(P_\Lambda(\lambda),(\tilde{\bm{E}}^{\lambda},\tilde{\bm{\omega}}^{\lambda}),\tilde{q}^\lambda,\tilde{t}^{\lambda},\tilde{z}_0^{\lambda},\tilde{z}_1^{\lambda}\right) = \left(0,({\bm{0}}^{\lambda},{\bm{0}}^{\lambda}),0,0,0,0\right)
    \end{align}
     is the only feasible point and this point does not contribute to the value of the objective function in Eq.~\eqref{App: OptProbCondEntSDPBound}. The same holds true for the value of objective function in Eq.~\eqref{App: SemiFinalOptProb} for
     \begin{align}
         \left(P_\Lambda(\lambda),(\bm{E}^{\lambda},\bm{\omega}^{\lambda}),{q}^\lambda,{t}^{\lambda},{z}_0^{\lambda},{z}_1^{\lambda} \right) = \left(0,(\bm{E}^{\lambda},\bm{\omega}^{\lambda}),{q}^\lambda,{t}^{\lambda},{z}_0^{\lambda},{z}_1^{\lambda} \right)
     \end{align}
    where $(\bm{E}^{\lambda},\bm{\omega}^{\lambda}),{q}^\lambda,{t}^{\lambda},{z}_0^{\lambda},{z}_1^{\lambda}$ is any feasible choice.
\end{proof}

\end{document}